\newif\ifstoc
\stoctrue 
\stocfalse

\ifstoc
  \documentclass{sig-alternate-2013}
  \newfont{\mycrnotice}{ptmr8t at 7pt}
  \newfont{\myconfname}{ptmri8t at 7pt}

  \permission{Permission to make digital or hard copies of all or part
    of this work for personal or classroom use is granted without fee
    provided that copies are not made or distributed for profit or
    commercial advantage and that copies bear this notice and the full
    citation on the first page. Copyrights for components of this work
    owned by others than ACM must be honored. Abstracting with credit is
    permitted. To copy otherwise, or republish, to post on servers or to
    redistribute to lists, requires prior specific permission and/or a
    fee. Request permissions from permissions@acm.org.}
  \conferenceinfo{STOC'15,}{June 14--17, 2015, Portland, Oregon, USA.}
  \copyrightetc{Copyright \copyright~2015 ACM \the\acmcopyr}
  \crdata{978-1-4503-3536-2/15/06\ ...\$15.00.\\
    http://dx.doi.org/10.1145/2746539.2746607}

\else
  \documentclass[letterpaper,11pt]{article}
  \usepackage{typearea}
  \typearea{14}
  \usepackage[compact]{titlesec}
  \usepackage{theorem}
  \usepackage{fullpage}
  \usepackage{setspace}

\fi

\newcommand{\full}[1]{\ifstoc\else#1\fi}
\newcommand{\short}[1]{\ifstoc#1\fi}

\full{
\makeatletter
 \setlength{\parindent}{0pt}
 \addtolength{\partopsep}{-2mm}
 \setlength{\parskip}{5pt plus 1pt}
 \addtolength{\abovedisplayskip}{-3mm}
 \addtolength{\textheight}{20pt}
\makeatother
}

\usepackage{theorem,latexsym,graphicx}
\usepackage{amsmath,amssymb,enumerate}
\usepackage{xspace}
\usepackage{ifpdf}
\usepackage{shadow,shadethm,color}
\usepackage{verbatim}

\definecolor{Darkblue}{rgb}{0,0,0.4}
\definecolor{Brown}{cmyk}{0,0.61,1.,0.60}
\definecolor{Purple}{cmyk}{0.45,0.86,0,0}

\ifstoc
  \usepackage[colorlinks,citecolor=Brown,bookmarks=true]{hyperref}
\else
  \usepackage[backref,colorlinks,citecolor=Brown,bookmarks=true]{hyperref}
\fi
\newcommand{\lref}[2][]{\hyperref[#2]{#1~\ref*{#2}}}

\newtheorem{theorem}{Theorem}[section]
\newtheorem{definition}[theorem]{Definition}

\newtheorem{lemma}[theorem]{Lemma}
\newtheorem{fact}[theorem]{Fact}
\newtheorem{assumption}[theorem]{Assumption}

\newtheorem{claim}[theorem]{Claim}

\newtheorem{corollary}[theorem]{Corollary}

\ifstoc
\else
\newenvironment{proof}{

\noindent{\bf Proof:}}
{\hfill$\blacksquare$

}
\fi

\newenvironment{proofof}[1]{

\noindent{\bf Proof of {#1}:}}
{\hfill$\blacksquare$

}

\newcommand{\junk}[1]{}
\newcommand{\ignore}[1]{}

\def\floor#1{\lfloor #1 \rfloor}
\def\ceil#1{\lceil #1 \rceil}

\def\abs#1{\mathopen| #1 \mathclose|}   

\newcommand{\poly}{\operatorname{poly}}

\newcommand{\sse}{\subseteq}
\newcommand{\I}{{\mathcal{I}}}
\newcommand{\G}{{\mathcal{G}}}
\newcommand{\e}{\varepsilon}

\newcommand{\ts}{\textstyle}

\newcommand{\bs}[1]{\boldsymbol{#1}}

\newcounter{note}[section]

\newcommand{\initOneLiners}{%
    \setlength{\itemsep}{0pt}
    \setlength{\parsep }{0pt}
    \setlength{\topsep }{0pt}
}
\newenvironment{OneLiners}[1][\ensuremath{\bullet}]
    {\begin{list}
        {#1}
        {\initOneLiners}}
    {\end{list}}

\newcommand{\initTwoLiners}{%
    \setlength{\itemsep}{1pt}
    \setlength{\parsep }{1pt}
    \setlength{\topsep }{1pt}
}
\newenvironment{TwoLiners}[1][\ensuremath{\bullet}]
    {\begin{list}
        {#1}
        {\initTwoLiners}}
    {\end{list}}

\newcommand{\alg}{\ensuremath{\mathsf{alg}}}
\newcommand{\sdp}{\ensuremath{\mathsf{sdp}}}
\newcommand{\SAplus}{\ensuremath{SA^+}\xspace}
\newcommand{\SA}{\ensuremath{SA}\xspace}

\newcommand{\cc}{\overline{\chi}}

\newcommand{\Exp}{\mathbb{E}}

\begin{document}

\ifstoc
\else

\title{On the Lov\'{a}sz Theta function for Independent
  Sets in Sparse Graphs}

\author{Nikhil Bansal\thanks{Eindhoven University of Technology. Email: n.bansal@tue.nl. Supported by NWO grant 639.022.211 and an ERC consolidator grant 617951.} \and Anupam Gupta\thanks{Computer Science Department, Carnegie Mellon
    University, Pittsburgh, PA 15213, USA. Research partly supported by
    NSF awards CCF-1016799 and CCF-1319811.} \and Guru
  Guruganesh$^\dagger$ } 
\fi

\date{}

\maketitle

\full{
\thispagestyle{empty}
}

\begin{abstract} 
  We consider the maximum independent set problem on sparse graphs with
  maximum degree~$d$. We show that the integrality gap of the Lov\'asz
  $\vartheta$-function based SDP is 
  \[ \widetilde{O}(d/\log^{3/2} d). \] This improves on the previous
  best result of $\widetilde{O}(d/\log d)$, and almost matches the 
  integrality gap of $\widetilde{O}(d/\log^2 d)$
  recently shown for stronger SDPs, namely those obtained using
  $\poly\log(d)$ levels of the \SAplus semidefinite hierarchy.  The
  improvement comes from an improved Ramsey-theoretic bound on the
  independence number of $K_r$-free graphs for large values of $r$.

  We also show how to obtain an algorithmic version of the
  above-mentioned \SAplus-based integrality gap result, via a coloring
  algorithm of Johansson. The resulting approximation guarantee of
  $\widetilde{O}(d/\log^2 d)$ matches the best unique-games-based
  hardness result up to lower-order $\poly (\log\log d)$ factors.
\end{abstract}

\full{
\newpage
\setcounter{page}{1}
}

\section{Introduction}

Given a graph $G=(V,E)$, an independent set is a subset of vertices $S$
such that no two vertices in $S$ are adjacent.  The maximum independent
set problem is one of the most well-studied problems in algorithms and
graph theory, and its study has led to various remarkable developments
such as the seminal result of Lov\'{a}sz~\cite{Ltheta} in which he introduced
the $\vartheta$-function based on semidefinite programming, as well as
several surprising results in Ramsey theory and extremal combinatorics.

In general graphs, the problem is notoriously hard to approximate.
Given a graph on $n$ vertices, the best known algorithm is due to
Feige~\cite{Feige04}, and achieves an approximation ratio of
$\widetilde{O}(n/\log^3 n)$; here $\widetilde{O}(\cdot)$ suppresses some
$\log \log n$ factors. On the hardness side, a result of
H{\aa}stad~\cite{Hastad96} shows that no $n^{1-\e}$ approximation
exists for any constant $\e>0$, assuming NP $\not \subseteq$
ZPP. The hardness has been improved more recently to $n/\exp((\log
n)^{3/4+\e})$ by Khot and Ponnuswami~\cite{KhotP06}. 

In this paper, we focus on the case of bounded-degree graphs, with
maximum degree $d$. Recall that the na\"{\i}ve algorithm (that
repeatedly picks an arbitrary vertex $v$ and deletes its neighborhood)
produces an independent set of size at least $n/(d+1)$, and hence is a
$d+1$-approximation. The first $o(d)$-approximation was obtained by
Halld\'{o}rsson and Radhakrishnan~\cite{HR94}, who gave a $O(d /\log
\log d)$ guarantee, based on a Ramsey theoretic result of Ajtai et al.~\cite{AEKS81}. 
Subsequently, an $O(d\,\frac{\log \log d}{\log d})$-approximation was 
obtained independently by several
researchers~\cite{AlonK98,Halperin02,Hall99} using the ideas of Karger,
Motwani and Sudan~\cite{KMS98} to round the natural SDP for the problem,
which was itself based on the Lov\'asz $\vartheta$-function.

On the negative side, Austrin, Khot and Safra~\cite{AKS11} showed an
$\Omega(d/\log^2 d)$ hardness of approximation, assuming the Unique
Games Conjecture. Assuming P $\neq$ NP, a hardness of $d/\log^4d$ was
recently shown by Chan~\cite{Chan}.  We remark that these hardness
results only seem to hold when $d$ is a constant or a very mildly
increasing function of $n$. In fact for $d=n$, the $\Omega(d/\log^2 d)$
hardness of~\cite{AKS11} is inconsistent with the known $O(n/\log^3 n)$
approximation~\cite{Feige04}. Hence throughout this paper, it will be
convenient to view $d$ as being a sufficiently large but fixed constant.

 
Roughly speaking, the gap between the $\Omega(d/\log^2 d)$-hardness and
the $\widetilde{O}(d/\log d)$-approximation arises for the following
fundamental reason. Approaches based on the SDP work extremely well if
the $\vartheta$-function has value more than $\widetilde{O}(n/\log d)$,
but not below this threshold. In order to to show an $\Omega(d/\log
d)$-hardness result, at the very least, one needs an instance with SDP
value around $n/\log d$, but optimum integral value about $n/d$.  While
graphs with the latter property clearly exist (e.g., a graph consisting
of $n/(d+1)$ disjoint cliques $K_{d+1}$), the SDP value for such graphs
seems to be low. In particular, having a large SDP value imposes
various constraints on the graph (for example, they cannot contain many
large cliques) which might allow the optimum to be non-trivially larger
than $n/d$, for example due to Ramsey-theoretic reasons.

Recently, Bansal~\cite{Ban15} leveraged some of these ideas to improve
the approximation guarantee by a modest $O(\log \log d)$ factor to
$d/\log d$ using polylog$(d)$ levels of the \SAplus hierarchy. His
improvement was based on combining properties of the \SAplus hierarchies
together with the ideas of~\cite{AEKS81}. He also showed that the
$O(\log^4 d)$-level \SAplus relaxation has an \emph{integrality gap} of
$\widetilde{O}(d/\log^2 d)$, where $\widetilde{O}(\cdot)$ suppresses
some $\log \log d$ factors.  The main observation was that as the
\SAplus relaxation specifies a local distribution on independent sets,
and if the relaxation has high objective value then it must be that any
polylog$(d)$ size subset of vertices $X$ must contain a large
independent subset. One can then use a result of Alon~\cite{Alon96}, in turn  based on an elegant entropy-based approach of
Shearer~\cite{Shearer95}, to show that such graphs have non-trivially
large independents sets. However, this argument is non-algorithmic; it
shows that the lifted SDP has a small integrality gap, but does not give
a corresponding approximation algorithm with running time
sub-exponential in $n$. This leads to the question whether this approach
can be converted into an approximation algorithm that outputs a set of
size $\widetilde{\Omega}(\log^2 d/d)$ times the optimal independent set,
or if there is a gap between the approximability and estimability of
this problem (as recently shown for an NP problem by Feige and
Jozeph~\cite{FJ14}).



\subsection{Our Results.}

Our results resolve some of these questions. For our first result, we
consider the standard SDP relaxation for independent set (without
applying any lift-and-project steps) and show that it is surprisingly
more powerful than the guarantee given by Alon and Kahale~\cite{AlonK98}
and Halperin~\cite{Halperin02}.

\begin{theorem}
  \label{th:sdp}
  On graphs with maximum degree $d$, the standard
  $\vartheta$-function-based SDP formulation for the independent set
  problem has an integrality gap of $\widetilde{O}(d/\log^{3/2}
  d)$.\footnote{Here and subsequently, $\widetilde{O}(\cdot)$ suppresses
    $\poly(\log \log d)$ factors.}
\end{theorem}

The proof of Theorem~\ref{th:sdp} is non-constructive; while it shows
that the SDP value is within the claimed factor of the optimal IS size,
it does not give an efficient algorithm to find such an approximate
solution. Finding such an algorithm remains an open question.

The main technical ingredient behind Theorem~\ref{th:sdp} is the
following new Ramsey-type result about the existence of large
independent sets in $K_r$-free graphs. This builds on a long line of
previous results in Ramsey theory (some of which we discuss in
Section~\ref{s:prel}), and is of independent interest.
(Recall that $\alpha(G)$ is the maximum independent set size in $G$.)

\begin{theorem}
  \label{th2} 
  For any $r > 0$, if $G$ is a $K_r$-free graph with maximum degree $d$
  then
  \begin{equation}
	\label{ramsey-term}
    \alpha(G) = \Omega\left(\frac{n}{d} \cdot\max\left( \frac{\log d}{r \log
          \log d}, \left(\frac{\log d}{\log
            r}\right)^{1/2}\right)\right).
  \end{equation}
\end{theorem}

Previously, the best known bound for $K_r$-free graphs was
$\Omega(\frac{n}{d}\, \frac{\log d}{r \log \log d})$ given by
Shearer~\cite{Shearer95}. Observe the dependence on $r$: when $r \geq
\frac{\log d}{\log \log d}$, i.e., when we are only guaranteed to
exclude very large cliques, this result does not give anything better
than the trivial $n/d$ bound.  It is in this range of $r \geq \log d$
that the second term in the maximization in \eqref{ramsey-term} starts to perform better and
give a non-trivial improvement. In particular, if $G$ does not contain
cliques of size $r = O(\log^{3/2}d)$ (which will be the interesting case
for Theorem~\ref{th:sdp}), Theorem~\ref{th2} gives a bound of
$\widetilde{\Omega}(\frac{n}{d}\,(\log d)^{1/2})$. Even for
substantially larger values such as $r = \exp(\log^{1-2\e} d)$, this
gives a non-trivial bound of $\widetilde{O}(\frac{n}{d}\, \log^\e d)$.

Improving on Shearer's bound has been a long-standing open problem in
the area, and it is conceivable that the right answer for $K_r$-free
graphs of maximum degree $d$ is $\alpha(G) \geq \frac{n}{d} \frac{\log
  d}{\log r}$. This would be best possible, since in Section~\ref{sec:lbd} we
give a simple construction showing a upper bound of $\alpha(G) =
O(\frac{n}{d} \frac{\log d}{\log r})$ for $r \geq \log d$, which to the
best of our knowledge is the smallest upper bound currently known.  The gap
between our lower bound and this upper bound remains an intriguing one
to close; in fact it follows from our proof of Theorem~\ref{th:sdp} that such a lower bound would imply an $\widetilde{O}(d/\log^2d)$ integrality gap for the standard SDP. Alon~\cite{Alon96} shows that this bound is achievable under
the stronger condition that the neighborhood of each vertex is
$(r-1)$-colorable.


We then turn to the approximation question. Our third result shows how
to make Bansal's result  algorithmic, thereby resolving the
approximability of the problem (up to lower order $\poly(\log \log d)$ factors), at
least for moderate values of~$d$.

\begin{theorem}
  \label{th:alg}
  There is an $\widetilde{O}(d /\log^2 d)$-approximation algorithm with
  running time\footnote{ While a $d$-level \SAplus relaxation has size $n^{O(d})$ in general, our relaxation
only uses variables corresponding to subsets of vertices that lie in the neighborhood of some vertex $v$, and thus has  $ n\cdot 2^{O(d)}$ variables}.
 $\poly(n)\cdot 2^{O(d)}$, based on rounding a $d$-level
  \SAplus semidefinite relaxation.
\end{theorem}

The improvement is simple, and is based on bringing the right tool to
bear on the problem. As in~\cite{Ban15}, the starting point is the
observation that if the $d$-level \SAplus relaxation has objective value
at least $n/s$ (and for Theorem~\ref{th:alg} the value $s=\log^2 d$
suffices), then the neighborhood of every vertex in the graph is
$k$-colorable for $k = s \cdot \textrm{polylog}(d)$ --- they are
``locally colorable''. By Alon's result mentioned above, such graphs
have $\alpha(G)=\Omega(\frac{n}{d}\, \frac{\log d}{\log k})$. However,
instead of using~\cite{Alon96} which relies on Shearer's entropy based
approach, and is not known to be constructive, we use an ingenious and
remarkable (and stronger) result of Johansson~\cite{Joh}, who shows that
the list-chromatic number of such locally-colorable graphs is
$\chi_{\ell}(G) = O(d \frac{\log k}{\log d})$. His result is based on a
very clever application of the R\"odl ``nibble'' method, together with
Lov\'asz Local Lemma to tightly control the various parameters of the
process at every vertex in the graph. Applying Johansson's result to our
problem gives us the desired algorithm.

Unfortunately, Johansson's preprint (back from 1996) was never
published, and cannot be found on the Internet.\footnote{We thank Alan
  Frieze for sharing a copy with us.}  For completeness (and to
facilitate verification), we give the proof in its entirety in the
Appendix. We essentially follow his presentation, but streamline some
arguments based on recent developments such as concentration bounds for
low-degree polynomials of random variables, and the algorithmic version
of Local Lemma. His manuscript contains many other results that build
upon and make substantial progress on a long line of work (we give more
details in Section~\ref{s:prel}).  We hope that this will make
Johansson's ideas and results accessible to a wider
audience. (Johansson's previous preprint~\cite{Joh-k3} showing the
analogous list-coloring result for triangle-free graphs is also
unavailable publicly, but is presented in the graph coloring book by
Molloy and Reed~\cite{MR02}, and has received considerable attention
since, both in the math~\cite{Alon99,Vu2002,Frieze2013} and computer
science communities~\cite{GrableP00,CPS14}.)

The proof of Theorem~\ref{th:alg} also implies the following new results
about the LP-based Sherali-Adams (\SA) hierarchies, without any SDP
constraints.
\begin{corollary}
  \label{cor:lp}
  The LP relaxation with clique constraints on sets of size up to $\log d$
  (and hence the relaxation $\SA_{(\log d)}$) has an integrality gap of
  $\widetilde{O}(d/\log d)$.  Moreover, the relaxation $\SA_{(d)}$ can
  be used to find an independent set achieving an $\widetilde{O}(d/\log
  d)$ approximation in time $\poly(n)\cdot2^{O(d)}$.
\end{corollary}

Since LP-based relaxations have traditionally been found to be very weak
for the independent set problem, it may be somewhat surprising that a
few rounds of the \SA-hierarchy improves the integrality gap by a
non-trivial amount.

All our results extend to the case when $d$ is the average degree of the
graph; by first deleting the (at most $n/2$) vertices with degree more
than $2\overline{d}$ and then applying the results.


\section{Preliminaries}
\label{s:prel}

Given the input graph $G=(V,E)$, we will denote the vertex set $V$ by
$[n]=\{1,\ldots,n\}$.  Let $\alpha(G)$ denote the size of a maximum
independent set in $G$, and $d$ denote the maximum degree in $G$.  The
naive greedy algorithm implies $\alpha(G) \geq n/(d+1)$ for every $G$.
As the greedy guarantee is tight in general (e.g., if the graph is a
disjoint union of $n/(d+1)$ copies of the clique $K_{d+1}$), the trivial
upper bound of $\alpha(G) \leq n$ cannot give an approximation better
than $d+1$ and hence stronger upper bounds are needed. A natural bound
is the clique-cover number $\cc(G)$, defined as the minimum number of
vertex-disjoint cliques needed to cover $V$.  As any independent set can
contain at most one vertex from any clique, $\alpha(G) \leq \cc(G)$.

\medskip\noindent {\bf Standard LP/ SDP Relaxations.} In the standard LP
relaxation for the independent set problem, there is variable $x_i$ for
each vertex $i$ that is intended to be $1$ if $i$ lies in the
independent set and $0$ otherwise. The LP is the following:
\ifstoc
\else
    \begin{equation}
    \label{lp}
     \max \sum_i x_i,  \quad \textrm{s.t.}\quad   x_i + x_j \leq 1 \quad \forall (i,j)\in E,  \quad \textrm{and}  
    \quad  x_i \in [0,1] \quad \forall i\in [n].
    \end{equation}
\fi

Observe that this linear program is very weak, and cannot give an
approximation better than $(d+1)/2$: even if the graph consists of
$n/(d+1)$ copies of $K_{d+1}$, the solution $x_i=1/2$ for each $i$ is a
feasible one.

In the standard SDP relaxation, there is a special unit vector $v_0$ (intended to indicate $1$) and a vector $v_i$ for each vertex $i$. The vector $v_i$ is intended to be $v_0$ if $i$ lies in the independent set and be $\bf{0}$ otherwise. This gives the following relaxation:
\ifstoc
\else
  \begin{equation}
  \label{sdp}
  \max \sum_i v_i \cdot v_0,  \quad \textrm{s.t.}\quad  v_0 \cdot v_0 =1, \quad v_0 \cdot v_i=v_i \cdot v_i \quad \forall i\in [n], \quad \textrm{and}  \quad   v_i\cdot v_j =0 \quad \forall (i,j)\in E.
  \end{equation}
\fi

Let $Y$ denote the $(n+1) \times (n+1)$ Gram matrix  with entries $y_{ij}= v_i \cdot v_j$, for $i,j\in\{0,\ldots,n\}$. Then we have the equivalent relaxation
\ifstoc
\else 
  \begin{equation}
  \label{gram}
  \max \sum_i y_{0i}, \quad \textrm{s.t.}\quad y_{00} = 1, \quad  y_{0i} = y_{ii}\quad \forall i\in [n], \quad  
  y_{ij} = 0 \quad \forall  (i,j) \in E \quad \textrm{and} \quad Y \succeq 0.
  \end{equation}
\fi

The above SDP which is equivalent to 
the well-known $\vartheta$-function of Lov\'{a}sz \cite{Laurent-notes}(Lemma 3.4.4), satisfies $\alpha(G)
\leq \vartheta(G) \leq \cc(G)$.  The $O(d \frac{\log \log d}{\log d})$
approximations due to \cite{AlonK98,Halperin02,Hall99} are all based on
SDPs. 

We will use the following important result due to Halperin
\cite{Halperin02} about the performance of the SDP. The form below
differs slightly from the one in \cite{Halperin02} as he works with a
$\{-1,1\}$ formulation. A proof for the form below can be found in
\cite[Theorem~3.1]{Ban15}. 
\begin{theorem} [Halperin \cite{Halperin02}, Lemma 5.2]
\label{th:halp}  Let $\eta \in [0,\frac{1}{2}]$ be a parameter and let 
$Z$ be the collection of vectors $v_i$ satisfying $\|v_i\|^2 \geq \eta$ in the SDP solution. Then there is an algorithm that returns an independent set of 
size $\Omega \left( \frac{d^{2\eta}}{d \sqrt{\ln d}}   |Z|\right)$.
\end{theorem}

Note that if $\eta = c \log \log d/ \log d$, then for $c\leq 1/4$
Theorem~\ref{th:halp} does not return any non-trivial independent
set. On the other hand, for $c \geq 1/4 $ the size of the independent
set returned rises exponentially fast with $c$.

For more details on SDPs, and the Lov\'{a}sz $\vartheta$-function, we
refer the reader to \cite{GLS88,MatG12}. 

\medskip\noindent{\bf Lower Bounds on the Independence Number.} 
As SDPs can handle cliques, looking at $\vartheta(G)$ naturally leads to
Ramsey theoretic considerations.  
In particular, if $\vartheta(G)$ is small then the trivial $n/(d+1)$
solution already gives a good approximation. Otherwise, if
$\vartheta(G)$ is large, then this essentially means that there are no
large cliques and one must argue that a large independent set exists
(and can be found efficiently). 

For bounded degree graphs, a well-known result of this type is that
$\alpha(G) = \Omega(n \frac{\log d}{d})$ for triangle-free graphs
\cite{AKS80,Shearer83} (i.e.~if there are no cliques of size $3$). A
particularly elegant proof (based on an idea due to Shearer
\cite{Shearer95}) is in \cite{AS}. Moreover this bound is tight, and
simple probabilistic constructions show that this bound cannot be
improved even for graphs with large girth.  


For the case of $K_r$-free graphs with $r \geq 4$, the situation is less
clear.  Ajtai et al.~\cite{AEKS81} showed that $K_r$-free graphs have
$\alpha(G)=\Omega(n (\log (\log d/r))/d)$, which implies that $\alpha(G)
= \Omega(n\log \log d/d)$ for $r\ll \log d$.  This result was the basis
of the $O(d/\log \log d)$ approximation due to
\cite{HR94}. Shearer~\cite{Shearer95} improved this result substantially
and showed that $\alpha(G) \geq \Omega(\frac1r \frac{n}{d} \frac{\log d}{\log \log d}) $
for $K_r$-free graphs. His result is based on an elegant entropy based
approach that has subsequently found many applications. However, it is
not known how to make this method algorithmic.
Removing the $\log \log d$ factor above is a major open question, even
for $r=4$. Also, note that his bound is trivial when $r \geq \frac{\log
  d}{\log \log d}$.

Interestingly, this result also implies another (non-algorithmic) proof
that the SDP has integrality gap $d \frac{\log \log d}{\log d}$.  In
particular, if the SDP objective is about $n/r$ this essentially implies
that the graph is $K_r$-free (as roughly each vertex contributes about 
$x_i=1/r$).  Thus, by Shearer's bound the integrality gap is
$(n/r)/\alpha(G) \leq d \frac{\log \log d}{\log d}$. It is interesting
to note that both Halperin's approach and Shearer's bound seem to get
stuck at the same point. 

Alon~\cite{Alon96} generalized the triangle-free result in a different
direction, also using the entropy method. He considered locally
$k$-colorable graphs, where the neighborhood of every vertex is
$k$-colorable and showed that $\alpha(G) = \Omega\left(\frac{n}{d}
  \frac{\log d}{\log k + 1} \right)$. Note that triangle-free graphs are
locally $1$-colorable.  This result also holds under weaker conditions,
and plays a key role in the results of \cite{Ban15} on bounding the
integrality gap of \SAplus relaxations.

\medskip\noindent{\bf Bounds on the chromatic number.}
Most of the above results generalize to the much more demanding setting
of list coloring.  All of them are based on ``nibble'' method, but
require increasingly sophisticated ideas. The intuition for why
$O(d/\log d)$ arises can be seen via a coupon-collector argument: if each vertex
in the neighborhood $N(v)$ chooses a color from $s$ colors independently
and u.a.r., they will use up all $s$ colors unless $d \leq O(s \log s)$,
or $s \geq \Omega(d/\log d)$. (Of course, the colors at the neighbors are
not chosen uniformly or independently.)  Kim showed that $\chi_\ell(G) =
O(d/\log d)$ for graphs with girth at least~$5$ \cite{Kim95}. His idea
was that for any $v$, and $u,w \in N(v)$, $N(u) \cap N(w) = \{v\}$
because of the girth, and hence the available colors at $u,w$ evolve
essentially independently, and hence conform to the intuition.

These ideas fail for triangle-free graphs (of girth $4$): we could have
a vertex $v$, with $u,w \in N(v)$, and $N(u) = N(w)$ (i.e., all their
neighbors are common). In this case the lists of available colors at $u$
and $w$ are far from independent: they would be \emph{completely
  identical}. Johansson~\cite{Joh-k3} had the crucial insight that this
positive correlation is not a problem, since there is no edge between
$u$ and $w$ (because of triangle-freeness!).  His clever proof
introduced the crucial notions of entropy and energy to capture and
control the positive correlation along edges in such $K_3$-free graphs.



If there are triangles, say if the graphs are only locally
$k$-colorable, then using these ideas na\"{\i}vely fails. A next key new
idea, also introduced by Johansson~\cite{Joh}, is to actually modify the
standard nibble process by introducing a {\em probability reshuffling}
step at each vertex depending on its local graph structure, which makes
it more complicated. In
Section~\ref{sec:jojo}, we give his result for locally-colorable graphs
in its entirety.


\medskip
\noindent {\bf Lift-and-project Hierarchies.} An excellent introduction
to hierarchies and their algorithmic uses can be found in
\cite{CT,L03}. 
Here, we only describe here the most basic facts that we need.  

The Sherali-Adams (\SA) hierarchy defines a hierarchy of linear programs
with increasingly tighter relaxations. At level $t$, there is a variable
$Y_S$ for each subset $S\subseteq [n]$ with $|S|\leq t+1$. Intuitively,
one views $Y_S$ as the probability that all the variables in $S$ are set
to $1$. Such a solution can be viewed as specifying a local distribution
over valid $\{0,1\}$-solutions for each set $S$ of size at most $t+1$.
A formal description of the $t$-round Sherali-Adams LP $\SA_{(t)}$ for
the independent set problem can be found in \cite[Lemma 1]{CT}.






For our purposes, we will also impose the PSD constraint on the
variables $y_{ij}$ at the first level (i.e., we add the constraints in
\eqref{gram} on $y_{ij}$ variables). We will call this the $t$-level \SAplus
formulation and denote it by $\SAplus_{(t)}$.  
To keep the notation consistent with the LP~(\ref{lp}), we will use $x_i$
to denote the marginals $y_{ii}$ on singleton vertices. 

\section{Integrality Gap}
\label{sec:intgap-theta}

In this section, we show Theorem~\ref{th:sdp}, that the integrality gap
of the standard Lov\'asz $\vartheta$-function based SDP relaxation is 
\[ \ts O\big(d \big(\frac{\log \log d}{\log d}\big)^{3/2}\big) =
\widetilde{O}\big(d/\log^{3/2} d\big). \] To show this we prove the
following result (which is Theorem~\ref{th2}, restated):

\begin{theorem}
  \label{th:2}
  Let $G$ be a $K_r$-free graph with maximum degree
  $d$. Then $$\alpha(G) = \Omega\left(\frac{n}{d} \max\left( \frac{\log
        d}{r \log \log d}, \left(\frac{\log d}{\log
          r}\right)^{1/2}\right)\right).$$ In particular, for $r =
  \log^{c}d$ with $c\geq 1$, we get $\alpha(G) = \Omega\big(\frac{n}{d}
    \big(\frac{\log d}{c\, \log \log d}\big)^{1/2}\big).$
\end{theorem}

We need the following basic facts. The first follows from a simple
counting argument (see~\cite[Lemma~2.2]{Alon96} for a proof).

\begin{lemma}
  \label{l:al}
  Let $F$ be a family of $2^{\e x}$ distinct subsets of an
  $x$-element set $X$. Then the average size of a member of $F$ is at
  least $ \e x/ (10 \log(1+ 1/\e))$.
\end{lemma}


\begin{fact}
  \label{f1}
  Let $G$ be a $K_r$-free graph on $x$ vertices, then 
  \begin{gather*}
    \alpha(G) \geq \max\left(\frac{x^{1/r}}{2}, \frac{\log x}{\log
        (2r)}\right).
  \end{gather*}
  Note that the latter bound is stronger when $r$ is large, i.e.,
  roughly when $r \geq \log x/\log \log x$.
\end{fact}
\begin{proof}
  Let $R(s,t)$ denote the off-diagonal $(s,t)$-Ramsey number, defined as
  the smallest number $n$ such that any graph on $n$ vertices contains
  either an independent set of size $s$ or a clique of size~$t$.

  It is well known that $R(s,t) \leq \binom{s+t-2}{s-1}$ \cite{ES35}.
  Approximating the binomial gives us the bounds $R(s,t) \leq (2s)^t$
  and $R(s,t) \leq (2t)^s$; the former is useful for $t \leq s$ and the
  latter for $s \leq t$. If we set $R(s,t) = x$ and $t=r$, the first
  bound gives $s \geq (1/2) x^{1/r}$ and the second bound gives $s \geq
  \log x/\log (2r)$.
\end{proof}


\vspace{2mm}

We will be interested in lower bounding the number of independent sets
$\I$ in a $K_r$-free graph.  Clearly, $\I \geq 2^{\alpha( G)}$ (consider
every subset of maximum independent set). However the following improved
estimate will play a key role in Theorem~\ref{th:2}. Roughly speaking it
says that if $\alpha(G)$ is small, in particular of size logarithmic in
$x$, then the independent sets are spread all over $G$, and hence their
number is close to $x^{\Omega(\alpha(G))}$.

\begin{theorem}
  \label{bound:i}
  Let $G$ be a $K_r$-free graph on $x$ vertices, and let $\I$ denote the
  number of independent sets in $G$. Then we have
  \begin{gather*}
    \log \I \geq \max\left( \frac{x^{1/r}}{2}, \frac{\log^2 x}{18 \log
        2r}\right).
  \end{gather*}
\end{theorem}
\begin{proof}
  The first bound follows trivially from Fact~\ref{f1}, and hence we
  focus on the second bound. Also, assume $r \geq 3$ and $x \geq 64$
  else the second bound is trivial.

  Define $s := \log x/\log (2r)$.  Let $G'$ be the graph obtained by
  sampling each vertex of $G$ independently with probability
  $p:=2/x^{1/2}$. The expected number of vertices in $G'$ is $px =
  2x^{1/2}$. Let $\G$ denote the good event that $G'$ has at least
  $x^{1/2}$ vertices. Clearly, $Pr[\G] \geq 1/2$ (in fact it is
  exponentially close to $1$).  Since the graph $G'$ is also $K_r$-free,
  conditioned on the event $\G$, it has an independent set of size at
  least $\log (x^{1/2})/\log (2r) = s/2$. Thus the expected number of
  independent sets of size $s/2$ in $G'$ is at least $1/2$.

  Now consider some independent set $Y$ of size $s/2$ in $G$. The
  probability that $Y$ survives in $G'$ is exactly $p^{s/2}$.  As the
  expected number of independent sets of size $s/2$ in $G'$ is at least
  $1/2$, it follows that $G$ must contain at least $(1/2) (1/p^{s/2}) $
  independent sets of $s/2$. This gives us that
  \begin{gather*}
    \log \I \geq \frac{s}{2} \log \left(\frac{1}{p}\right) - 1 \geq
    \frac{s}{2} \log x^{1/2} - \frac{s}{2} -1 \geq \frac{s}{18} \log x,
  \end{gather*}
  where the last inequality assumes that $x$ is large enough. 
\end{proof}

We are now ready to prove Theorem~\ref{th:2}.

\begin{proof}
  We can assume that $d \geq 16$, else the claim is trivial.  Our
  arguments follow the probabilistic approach
  of~\cite{Shearer95,Alon96}. Let $W$ be a random independent set of
  vertices in $G$, chosen uniformly among all independent sets in $G$.
  For each vertex $v$, let $X_v$ be a random variable defined as $X_v =
  d |{v} \cap W| + |N(v)\cap W|$.

  Observe that $|W|$ can be written as $\sum_v |{v} \cap W| $; moreover,
  it satisfies $|W| \geq (1/d) \sum_v |N(v)\cap W|$, since a vertex in
  $W$ can be in at most $d$ sets $N(v)$. Hence we have that 
  \begin{gather*}
    |W| \geq \frac{1}{2d} \sum_v X_v.
  \end{gather*}
  Let $\gamma = \max\big( \frac{\log d}{r \log \log d}, \big(\frac{\log d}{\log
    r}\big)^{1/2}\big)$ denote the improvement factor in
  Theorem~\ref{th:2} over the trivial bound of $n/d$.  Thus to show that
  $\alpha(G)$ is large, it suffices to show that
  \begin{equation}
    \label{eq13}
    \Exp[X_v] \geq c \gamma
  \end{equation}
  for each vertex $v$ and some fixed constant $c$.

  In fact, we show that \eqref{eq13} holds for every conditioning of the
  choice of the independent set in $V -(N(v) \cup \{v\})$. In
  particular, let $H$ denote the subgraph of $G$ induced on $V -(N(v)
  \cup \{v\})$. For each possible independent set $S$ in $H$, we will
  show that
  \begin{gather*}
    \Exp[X_v \mid W \cap V(H) =S] \geq c \gamma.
  \end{gather*}
  Fix a choice of $S$. Let $X$ denote the non-neighbors of $S$ in
  $N(v)$, and let $x = |X|$. Let $\e$ be such that $2^{\e
    x}$ denotes the number of independent sets in the induced subgraph
  $G[X]$.  Now, conditioning on the intersection $W \cap V (H) = S$,
  there are precisely $2^{\e x} + 1$ possibilities for W: one in
  which $W = S \cup \{v\}$, and $2^{\e x}$ possibilities in which
  $v \notin W$ and $W$ is the union of $S$ with an independent set in
  $G[X]$.

  By Lemma~\ref{l:al}, the average size of an independent set in $X$ is
  at least $\frac{\e x}{10 \log 1/\e + 1}$ and thus we have
  that
  \begin{equation}
    \label{eq:al1} 
    \Exp[X_v \mid W \cap V(H) =S]
    \geq    d \frac{1}{2^{\e x} +1 }  + \frac{\e x}{10 \log
      (1/\e +1) } \frac{2^{\e x}} {2^{\e x} + 1}  
  \end{equation}
  Now, if $2^{\e x} +1 \leq \sqrt{d}$, then the first term is at
  least $\sqrt{d}$, and we've shown~\eqref{eq13} with room to spare.  So
  we can assume that $\e x \geq (1/2) \log d$. Moreover, by
  Theorem~\ref{bound:i},
  \begin{gather*}
    \e x \geq \max \bigg( \frac{x^{1/r}}{2}, \frac{\log^2 x}{18 \log (2r)} \bigg)
  \end{gather*}
  and hence the right hand side in (\ref{eq:al1}) is at least 
  \begin{gather}
    \frac{1}{40 \log (1/\e +1) } \max \left(\frac{\log d}{2},
      \frac{x^{1/r}}{2}, \frac{\log^2 x}{18 \log 2r} \right)  \notag \\
    \geq \frac{1}{40 \log (x +1) } \max \left(\frac{\log d}{2},
      \frac{x^{1/r}}{2}, \frac{\log^2 x}{18 \log 2r }  \right), 
    \label{eq:main}
  \end{gather}
  where the inequality uses $\e \geq 1/x$ (since $\e x
  \geq (1/2) \log d \geq 1$).

  First, let's consider the first two expressions in~(\ref{eq:main}).
  If $x \geq \log^r d$, then as $x^{1/r}/\log (x+1)$ is increasing in
  $x$, it 
  follows that the right hand side of \eqref{eq:main} is at least
  \begin{gather*}
    \frac { x^{1/r} }{80 \log (x+1) } = \Omega\left(\frac{\log d}{r \log
        \log d} \right).
  \end{gather*}
  On the other hand if $x \leq \log^r d$, then we have that the right
  hand side is again at least
  \begin{gather*}
    \frac{1}{40 \log (x+1)}  \frac{\log d}{2} = \Omega\left(\frac{\log d}{r \log
        \log d} \right).
  \end{gather*}
  Now, consider the first and third expressions in in~(\ref{eq:main}). 
  Using the fact that $\max(a,b) \geq \sqrt{ab}$ with $a = (\log d)/2$
  and $b= (\log^2 x)/(18 \log 2r)$, we get that (\ref{eq:main}) is at
  least $\Omega \left(\frac{\log d}{\log r}\right)^{1/2}$. Hence, for
  every value of $x$ we get that~(\ref{eq:main}) is at least
  $\Omega(\gamma)$ as desired in~(\ref{eq13}); this completes the proof of Theorem~\ref{th:2}.
\end{proof}
  


  

We can now show the main result of this section. 

\begin{theorem} 
  The standard SDP for independent set has an integrality gap of
  \[ O\left(d \left(\frac{\log \log d}{\log d}\right)^{3/2}\right). \]
\end{theorem}
\begin{proof}
  Given a graph $G$ on $n$ vertices, let $\beta \in [0,1]$ be such that
  the SDP on $G$ has objective value $\beta n$. If $\beta \leq
  2/\log^{3/2}d$, the na\"{\i}ve greedy algorithm already implies a
  $d/\log^{3/2}d$ approximation. Thus, we will assume that $\beta \geq
  2/\log^{3/2}d$.

  Let us delete all the vertices that contribute $x_i \leq \beta/2$ to
  the objective. The residual graph has objective value at least $\beta
  n - (\beta/2) n = \beta n/2$.

  Let $\eta = 2 \log \log d/ \log d$. If there are more than $n/\log^2 d$ vertices with $x_i \geq \eta$, 
  applying Theorem~\ref{th:halp} to the collection of these vertices
  already gives independent set of size at least
  \begin{gather*}
    \Omega\left( \frac{d^{2 \eta}}{d \sqrt{\ln d}} \cdot \frac{n}{\log^2
        d}\right) = \Omega\left(\frac{n \log^{3/2} d}{d}\right),
  \end{gather*}
  and hence a $O(d/\log^{3/2}d)$ approximation.
	
	Thus we can assume that fewer than $n/\log^2 d$ vertices have $x_i \geq \eta$. 
	As each vertex can contribute at most
  $1$ to the objective, the SDP objective on the residual graph obtained by deleting the vertices with $x_i \geq \eta$ is at
  least $\beta n/2 - n/(\log^2 d)$ which is at least $ \beta n/3$, since
  $\beta \geq 2/\log^{3/2} d$.

  So we have a feasible SDP solution on a subgraph $G'$ of
  $G$, where the objective is at least $\beta n/3$ (here $n$ is the
  number of vertices in $G$ and not $G'$) and each surviving vertex $i$
  has value $x_i$ in the range $[\beta/2, \eta]$. 

  As $x_i \leq \eta$ for each $i$, and the SDP objective
  is at least $\beta n/3$, the number of vertices $n'$ in $G'$ satisfies
  $n' \geq (\beta n/3)/\eta = \Omega(n\beta/\eta)$.  Moreover, as $x_i
  \geq \beta/2$ for each vertex $i\in G'$, and the SDP does not put more
  than one unit of probability mass on any clique, it follows that $G'$
  is $K_r$-free for $r=2/\beta = \log^{3/2} d$.  Applying
  Theorem~\ref{th:2} to $G'$ with parameter $r = \log^{3/2} d$, we obtain
  that $G'$ has an independent set of size
  \ifstoc
\else 
    \begin{gather*}
      \Omega\left(\frac{n'}{d} \sqrt{\frac{\log d}{\log r}}\right) =
      \Omega\left(\frac{n'}{d} \sqrt{\frac{\log d}{\log \log d}}\right)  
       = \Omega\left(\frac{n\,\beta}{d\,\eta} \sqrt{1/\eta} \right) =
      \Omega\left(\frac{\beta n}{d} \cdot \eta^{-3/2}\right).
    \end{gather*}
  \fi
  
  The SDP objective for $G$ was $\beta n$, so the integrality gap is
  $O(d \eta^{3/2}) = O(d (\frac{\log \log d}{\log d})^{3/2}) $.
\end{proof}

\subsection{An upper bound}
\label{sec:lbd}

We give a simple construction that $\alpha(G) \leq \frac{n}{d}
\frac{\log d}{\log r}$ for $r \geq \log d$.  We use the standard lower
bound $R(s,t) = \Omega(t^{s/2})$ for off-diagonal Ramsey numbers for $t
\geq s$. Setting $t=r$ with $r \geq \log d$, it follows that there exist
$K_r$-free graphs $H$ on $d$ vertices such that $\alpha(H) = O(\log
d/\log r)$. Now set $G$ to be $n/d$ disjoint copies of~$H$.

\ifstoc
\else
\section{An Algorithm using Lift-and-Project}
\fi
\label{sec:algo}

In this section, we briefly illustrate how to make Bansal's argument
about the integrality gap of the lifted SDP~\cite{Ban15} algorithmic.
Consider the $\SAplus_{(d)}$ relaxation on $G$, and let $\sdp(G)$ denote its
value.  We can assume that 
\begin{gather}
  \sdp(G) \geq n/\log^2 d, \label{eq:sdp-large}
\end{gather}
otherwise the naive
algorithm already gives a $d/\log^2 d$ approximation.

Let $\eta = 3 \log \log d/\log d$, and $Z$ denote the set of vertices
$i$ with $x_i \geq \eta$.  We can assume that $|Z| \leq n/(4 \log^2 d)$,
otherwise applying Theorem~\ref{th:halp} gives an independent set of
size $ \Omega(|Z| \cdot d^{2\eta}/(d \sqrt{\log d})) = \Omega(n\log^2
d/d)$. Applying Theorem~\ref{th:halp} is fine, since our solution 
belongs to \SAplus and hence is a valid SDP solution. Hence,
\begin{gather*}
  \sdp(G) \leq |Z|\cdot 1 + (n - |Z|)\cdot \eta \leq (n/(4 \log^2 d)) \cdot 1 + n \cdot \eta = 2 \eta n.
\end{gather*}

Let $V'$ denote the set of vertices $i$ with $x_i \in [1/(4\log^2 d), \eta]$.
\begin{claim}
  \label{large:v'}
  $|V'| \geq \sdp(G)/(2 \eta).$
\end{claim}
\begin{proof}
  The total contribution to $\sdp(G)$ of vertices $i$ with $x_i \leq
  1/(4 \log^2 d)$ can be at most $n/(4 \log^2d)$, which
  by~(\ref{eq:sdp-large}) is at most $\sdp(G)/4$. Similarly, the
  contribution of vertices in $Z$ is at most $|Z|$, which is again at
  most $\sdp(G)/4$. Together this gives $\sdp(G') \geq \sdp(G)/2$.  As each
  vertex in $V'$ has $x_i \leq \eta$, the claim follows.
\end{proof}

\begin{lemma}
  The graph $G' = G[V']$ induced on $V'$ is locally $k$-colorable for
  $k=O(\log^3 d)$.
\end{lemma}
\begin{proof}
  Consider the solution $\SAplus_{(d)}$ restricted to $G'$.  For a vertex $v
  \in V'$, let $N(v)$ denote its neighborhood in $G'$. As $|N(v)|\leq d$
  and $x_i \geq 1/(4 \log^2 d)$ for all $i \in N(v)$, the $\SAplus_{(d)}$
  solution defines a ``local distribution'' $\{X_S\}_{S \sse N(v)}$ over
  subsets of each neighborhood with the following properties:
  \begin{TwoLiners}
    \item[(i)] $X_S \geq 0$ and $\sum_{S \subset N(v)} X_S = 1$,  
    \item[(ii)] $X_S>0$ only if $S$ is independent in the subgraph
      induced on $N(v)$, and 
    \item[(iii)] for each vertex $i \in N(v)$, it holds that 
      \[x_i = \sum_{S \subseteq N(v): i\in S} x_S \geq 1/(4 \log^2 d).\]
  \end{TwoLiners}
  Scaling up the solution $X_S$ by $4 \log^2 d$ thus gives a valid
  fractional coloring of $N(v)$ using $4 \log^2 d$ colors, which by a
  set-covering argument implies that $\chi(N(v)) = O(\log^2 d\cdot \log N(v))
  = O(\log^3 d)$.
\end{proof}

Using Johansson's coloring algorithm for locally $k$-colorable graphs
(Theorem~\ref{thm:jo-local}) we can find an independent set of $G'$ with size
\[ \alg(G') = \Omega \left( \frac{|V'|}{d} \cdot \frac{\log d}{\log
    (k+1)} \right). \] Using $k = O(\log^3 d)$ and Claim~\ref{large:v'}
this implies an algorithm to find independent sets in degree~$d$ graphs, 
with an integrality gap of
\begin{gather*}
  \frac{\sdp(G)}{\alg(G)} \leq \frac{\sdp(G)}{\alg(G')} \leq
  O\left(\frac{d \eta \log (k+1)}{\log d}\right) =
  \widetilde{O}\left(\frac{d}{\log^2 d}\right).
\end{gather*}
Our algorithm only required a fractional coloring on the neighborhood of vertices.
Since they are at most $2^d$ independent sets in each neighborhood,
there are at most $n\cdot 2^d$ relevant variables in our SDP. Hence, we can compute 
the relevant fractional coloring in time $\poly(n)\cdot2^{O(d)}$.

\section{LP-based guarantees}
\label{sec:lp-based}

We prove Corollary \ref{cor:lp}.  Consider the standard LP~\eqref{lp}
strengthened by the clique inequalities $\sum_{i \in C} x_i \leq 1$ for
each clique $C$ with $|C|\leq \log d$.  As each clique lies in the
neighborhood of some vertex, the number of such cliques is at most $n
\cdot \binom{d}{\log d}$.  Let $\beta n$ denote the objective value of
this LP relaxation.  We assume that $\beta \geq 2/\log d$, otherwise the
na\"{\i}ve algorithm already gives a $d/\log d$ approximation.

Let $B_0$ denote the set of vertices with $x_i \leq 1/\log d = \beta/2$.
For $j=1,\ldots,k$, where $k=\log \log d$, let $B_j$ denote the set of
vertices with $x_i \in (2^{j-1}/\log d, 2^j/\log d]$.  Note that
$\sum_{j \geq 1} \sum_{i \in B_j} x_i = \beta n - \sum_{i \in B_0} x_i
\geq \beta n/2$, and thus there exists some index $j$ such that $\sum_{i
  \in B_j} x_i \geq \beta n/(2k)$.

Let $\gamma = 2^{j-1}/\log d$; for each $i \in B_j$, $x_i \in (\gamma,
2\gamma]$. Since $x_i > \gamma$ for each $i\in B_j$, the clique
constraints ensure that the graph induced on $B_j$ is $K_r$-free for $r=
1/\gamma$. Moreover, since $x_i \leq 2 \gamma$ for each $i \in B_j$,
$|B_j| \geq \frac{1}{2\gamma} \cdot \frac{\beta n}{2k}$. By Shearer's
result for $K_r$-free graphs we obtain
$$\alpha(B_j)  = \Omega\left( |B_j| \cdot \frac{\gamma \log d}{d \log \log d} \right)=  \Omega\left( \frac{\beta n \log d} {d (\log \log d)^2  }\right).$$
This implies the claim about the integrality gap.

A similar argument implies the constructive result. Let $\beta n$ denote
the value of the $\SA_{(d)}$ relaxation.  As before, we assume that
$\beta \geq 2/\log d$ and divide the vertices into $1+ \log \log d$
classes. Consider the class $B_j$ with $j\geq 1$ that contributes most
to the objective, and use the fact that the graph induced on $B_j$ is
locally $k$-colorable for $k = (\log d/2^{j-1} \cdot \log d) = O(\log^2
d)$. As in Section~\ref{sec:algo}, we can now use Johansson's coloring
algorithm Theorem~\ref{thm:jo-local} to find a large
independent set.


%
%

\subsection*{Acknowledgments} 
We thank Alan Frieze for sharing a copy of the manuscript of Johansson
with us. We thank Noga Alon, Tom Bohman, Alan Frieze, and Venkatesan
Guruswami for enlightening discussions.

\ifstoc
\else
{\small
\bibliographystyle{alpha}
\bibliography{deg-d}
}
\fi

\short{\appendix}

\newcommand{\bee}{b}
\newcommand{\cee}{\hat{c}}
\newcommand{\pc}{p_c}
\newcommand{\pa}{p_a}

\newcommand{\phat}{\widehat{p}}
\newcommand{\pstar}{{p^\star}}
\newcommand{\crowd}{100 \ln \Delta}
\newcommand{\g}{\gamma}
\newcommand{\vg}{{(v,\gamma)}}
\newcommand{\ug}{{(u,\gamma)}}
\newcommand{\wg}{{(w,\gamma)}}
\newcommand{\whpD}{\ensuremath{\mathbf{whp}_\Delta}\xspace}
\newcommand{\Ew}{{E^w_{uv}(\gamma)}}

\newcommand{\one}[1]{\mathbf{1}_{(#1)}}
\newcommand{\nrg}{\xi}
\newcommand{\nrgy}[2]{\xi(#1, #2)}
\newcommand{\sprod}{\overline{\prod}}
\newcommand{\sProd}{\overline{\Prod}}
\newcommand{\morv}{{m.o.\ r.v.}\xspace}
\newcommand{\morvs}{{m.o.\ r.v.s}\xspace}
\newcommand{\khat}{\widehat{\kappa}}
\newcommand{\Mtilde}{\smash{\widetilde{M}}}

\newcommand{\B}[0]{{\ensuremath{\mathcal{B}}}\xspace}
\newcommand{\E}[0]{{\ensuremath{\mathcal{E}}}\xspace}
\newcommand{\F}[0]{{\ensuremath{\mathcal{F}}}\xspace}
\newcommand{\T}[0]{{\ensuremath{\mathcal{T}}}\xspace}
\newcommand{\C}[0]{{\ensuremath{\mathcal{C}}}\xspace}
\renewcommand{\P}[0]{{\ensuremath{\mathcal{P}}}\xspace}

\section{Johansson's Algorithm for Coloring Sparse Graphs}
\label{sec:jojo}

For completeness, \short{in the full version of the paper} we give proofs for two results of
Johansson~\cite{Joh} on coloring degree-$d$ graphs: one about graphs
where vertex neighborhoods can be colored using few colors
(``locally-colorable'' graphs), and another about $K_r$-free graphs.

\begin{theorem}
  \label{thm:jo-local}
  For any $r, \Delta$, there exists a randomized algorithm that, given a
  graph $G$ with maximum degree $\Delta$ such that the neighborhood of
  each vertex is $r$-colorable, outputs a proper coloring of $V(G)$
  using $O(\frac{\Delta }{\ln\Delta} \ln r)$ colors in expected
  $\poly(n 2^\Delta)$ time.
\end{theorem}

\begin{theorem}
  \label{thm:jo-kr}
  For any $r, \Delta$, there exists a randomized algorithm that, given a
  graph $G$ with maximum degree $\Delta$ which excludes $K_r$ as an
  subgraph, outputs a proper coloring of $V(G)$ using
  $O\big(\frac{\Delta }{\ln\Delta} (r^2 + r \ln \ln \Delta)\big)$ colors
  in expected $\poly(n)$ time.
\end{theorem}
 
We emphasize that Johansson's manuscript contains proofs of other
results and extensions, such as colorability under weaker conditions
than above, and extensions to list-coloring; we omit these extensions
for now. Our presentation largely follows his, but streamlines some of
the proofs using techniques that have developed since, such as
concentration bounds for low-degree polynomials of variables, and
dependent rounding techniques.  
\ifstoc
\else 
  Roadmap: we first give the intuition
  in Section~\ref{sec:overview}. We give the proof of
  Theorem~\ref{thm:jo-local} in \S\ref{sec:randproc}--\S\ref{sec:inv}, 
  and then show how to extend it to $K_r$-free graphs in
  \S\ref{sec:kr-coloring}.  
\fi

\subsection{Overview and Ideas}
\label{sec:overview}

Johansson's algorithm for locally-colorable graphs uses the ``nibble''
approach: in each round, some $\theta > 0$ fraction of vertices get
colored from their currently-allowable colors. The goal is to argue
(using concentration of measure, and the Local Lemma) that the degree of
each surviving vertex goes down exponentially like $(1 - O(\theta))^t$,
whereas the number of colors does not decrease very fast. This means
that after $\approx (\e/\theta) \ln \Delta$ rounds the degree of the
remaining vertices is be smaller than $\Delta^{1-\e}$ before running out
of the prescribed number of colors, at which point even the na\"ive
greedy algorithm can color the remaining vertices with a few more
colors. The proof for the degree reduction uses concentration bounds for
quadratic polynomials of random variables. The real challenge is to
lower-bound the number of remaining colors. Johansson's argument shows
that the entropy of the probability distribution of a vertex over its
colors remains high throughout the process, and hence there must be many
colors available. This requires a carefully orchestrated process, which
we describe next.
\short{\vfill\eject}

In more detail (but still at a high level): in each round, some $\theta
\approx \Delta^{-1/4}$ fraction of the vertices get activated, and each
tentatively chooses a color from its own probability distribution. (This
per-vertex distibution is initially the uniform distribution.) Any
vertex that gets the same color as its neighbor rejects its color; since
the number of these is small, we can ignore these for now. Then each
tentatively colored vertex (say $v$ with color $\g$), with probability
$\frac12$ accepts color $\g$ permanently and deletes the probability
mass corresponding to $\g$ from its neighbors (so that they cannot take
color $\g$); with the remaining probability $\frac12$, $v$ rejects color
$\gamma$ for this round and waits for another round. In order to ensure
the total probability mass at each vertex remains about $1$, since the
first option caused the probability mass for color $\g$ to decrease at
the neighbors, the second option must increase color $\g$'s mass at the
neighbors. If two of these neighbors $u, w$ are connected by an edge,
this means that we're increasing the chance that both these will get
color $\gamma$; this is potentially worrisome.

This problem does not arise if the graph is triangle-free, because there
are no edges in the neighborhood of any vertex. In this case Johansson's
previous preprint~\cite{Joh-k3} argued that the entropy of each vertex's
distribution remains high---itself a clever and delicate argument
(see~\cite[Chapters~12-13]{MR02}). However, if we just assume that the
graph is locally $r$-colorable, the existence of edges in node $v$'s
neighborhood means that the probability mass for a color at both
endpoints of an edge may become higher, creating undesirable positive
correlations.  What Johansson's new proof does is simple but ingenious:
it ``reshuffles'' the measure for the color randomly to some independent set
in the neighborhood. This is where the $r$-colorability condition kicks
in: since there are large independent sets (of size $\Delta/r$) in each
neighborhood, the reshuffling does not change the probabilities too
suddenly. Now carefully applied concentration bounds and LLL show a
similar behavior as in the triangle-free case, and proves
Theorem~\ref{thm:jo-local}.

The argument for $K_r$-free graphs requires a more involved
recursive reshuffling operation: in this case the size of the
independent sets may be too small (if we just use Ramsay's bound, for
instance), so the idea is to move the measure (on average) to sets that
avoid $K_t$ for $t$ smaller than $r$. This process (which Johansson
calls a ``trimming modifier'') creates a very slight negative
correlation on the edges, but this suffices to show
Theorem~\ref{thm:jo-kr}.

\ifstoc
\else
\subsection{Notation and Preliminaries}

We now define some notation and concepts, and give properties useful for
the following proofs.
We will interchangeably use $u \sim v$ and $u \in N(v)$ to denote that $u$ and $v$ are adjacent.

\subsubsection{Mean-One Random Variables}
\label{sec:morv}

An r.v.\ $X$ is a \emph{mean-one random variable} (\morv) if $X$ only
takes on values in $\{0\} \cup [1,\infty)$, and $\Exp[X] = 1$. One
simple class of \morvs take on some value $c \geq 1$ w.p.\ $1/c$, and
$0$ w.p.\ $1-1/c$.

\subsubsection{The Stopped Product}
\label{sec:stopped}

Given a sequence of non-negative random variables $Y_1, Y_2, \ldots,
Y_m$, and a ``threshold'' value $a \geq 0$, define a stopping time
$\tau_a$ as
\[ \ts \tau_a := \min \bigg\{ t \mid \prod_{i \leq t} Y_i \geq a
\bigg\} \] Then the \emph{stopped product} $\sprod_i Y_i$ is defined as
\[ \ts \sprod_i Y_i := \prod_{i \leq \min(m, \tau_a)} Y_i \]

\subsubsection{The \texorpdfstring{$\kappa$}{\kappa} and \texorpdfstring{$\widehat{\kappa}$}{\kappa} Functions}
\label{sec:kappa}

For a random variable $X$, define the function 
\begin{gather}
  \kappa(X) := \Exp[ X \ln X ]. 
\end{gather}
 The following facts are easy to verify, and will be useful in calculations.
\begin{TwoLiners}
\item[(a)] If $X, Y$ are independent, then $\kappa(XY) = \kappa(X)\,\Exp[Y]
  + \kappa(Y)\,\Exp[X]$.
\item[(b)] Hence if $X,Y$ are independent \morvs, then $\kappa(XY) =
  \kappa(X) + \kappa(Y)$.
\item[(c)] Also, $\kappa(aX) = \Exp[X](a \ln a) + a\kappa(X)$.
\item[(d)] For an event $\E$ and the associated \morv $X =
  \frac{\one{\E}}{\Pr[\E]}$, $\kappa(X) = \ln(1/\Pr[\E])$.
\item[(e)] For a stopped product $X = \sprod X_i$ (as defined in
  \S\ref{sec:stopped}) of independent \morvs
  with respect to some threshold $a$,
\[ \ts \kappa(X) \leq \sum_i
  \kappa(X_i). \] 
\item[(f)] If $X = (1 - \one{\E}) + \one{\E} \cdot Y$, and $Y$ is
  independent of the event $\E$, then $\kappa(X) = \Pr[\E]\,
  \kappa(Y)$.
\end{TwoLiners}
It is also useful to define $\khat(X)$ as an absolute upper
bound on $X$:
\begin{gather}
  \khat(X) := \inf \{ c \mid X \leq c~~ a.s. \}. \label{eq:jo-khat}
\end{gather}

\subsection{The Algorithm for Locally-Colorable Graphs}
\label{sec:randproc}

Let us present the algorithm for Theorem~\ref{thm:jo-local} about
finding colorings of $r$-locally-colorable graphs. The proof follows in
\S\ref{sec:jo-analysis}.

Let $s = O(\Delta \frac{\ln \Delta}{\ln r})$ be the number of colors
we are aiming for, and $L$ is the set of $s$ colors. For vertex $v$,
let $\P_v = \{v\} \times L$ be a collection of tuples indicating which
colors are still permissible for $v$. The term \whpD denotes ``with
probability at least $1 - 1/\poly(\Delta)$''. 

We follow the algorithmic outline from the overview n
\S\ref{sec:overview}.  The algorithm starts with each vertex $v$ having
a uniform probability distribution $p^0\vg = 1/s$ over the colors $\g
\in L$. In each stage $t$ we pick some vertices from the current graph
$G^t$ and color them based on the current values of $p^t\vg$, then
update the probability distributions of the other vertices to get
$p^{t+1}\vg$, drop the colored vertices to get $G^{t+1}$, and proceed to
the next stage. (This is the so-called ``nibble''.) The goal is to show
that after sufficiently many stages we have a partial proper coloring
using at most $s$ colors, and the degree of the graph induced by the
yet-uncolored vertices is $\Delta^{1-\e}$. We can then use a greedy
algorithm to color the remaining vertices.

The process in a generic stage $t$ is as follows (we drop superscripts
of $t$ to avoid visual clutter). 

\begin{enumerate}
\item Let $\phat \in (0,1)$ be a threshold to be defined later. Define
  \begin{align}
    \pa\vg &:= p\vg \cdot \one{p\vg \leq \phat} \\
    \pc\vg &:= \pa\vg \cdot \one{\sum_{u \sim v} p(u,\gamma) \leq
      \crowd}.
  \end{align}
  Hence $\pa$ zeroes out any (vertex, color) tuple $\vg$ which has a
  high value, and $\pc$ additionally zeroes out $\vg$ when $v$'s
  neighbors have a lot of probability mass on color~$\g$. \footnote{The
    definition of $\pc$ is not required to prove
    Theorem~\ref{thm:jo-local}, but is useful in extending the result to
    $K_r$-free graphs. The reader only interested in the former result
    should think of $\pc = \pa$ for this discussion.} 

\item For each vertex $v$ in the current graph $G$ and each color
  $\gamma$ in $L$, independently flip a coin with probability $\theta
  \pc\vg$. (The parameter $\theta \in (0,1)$ is defined later.) Let
  $A_\vg$ be this indicator variable. If $A_\vg = 1$ then color $\gamma$
  is \emph{tentatively assigned} to $v$. Many colors may be tentatively
  assigned to $v$.

  Also, let $\eta_\vg \sim \text{Bin}(1/2)$ be an unbiased coin-flip
  independent of all else.

\item For $v \in V$, let $\T_v = \{ \vg \mid A_\vg = 1\}$, and let
  $\T = \cup_v \T_v$ be all the tentatively assigned tuples. Let
  \[ \C_v = \{ \vg \mid A_\vg = 1 \land \eta_\vg = 1 \land A_\ug = 0 \,
  \forall u \sim v \}, \] and let $\C = \cup_v \C_v$ similarly. Note
  that the pair $\vg \in \T$ is dropped from $\C$ if any neighbor of $v$
  is tentatively assigned color~$\gamma$ (this ensures proper
  colorings), or if its own coin-flip $\eta_\vg$ comes up tails (this
  gives us a ``damping'' that is useful for the rest of the process).

\item For $v \in G$, if there exists some $\gamma$ such that $(v,
  \gamma) \in \C$, then color $v$ with an arbitrary such $\gamma$ and
  remove $v$ from $G'$. So the events $\{ v \in G' \} = \{ \C_v =
  \emptyset\}$.
\end{enumerate}

Now comes the changing of the probabilities via the ``modifiers'' (which
are simply mean-one r.v.s, as defined in \S\ref{sec:morv}).
\begin{enumerate}
\setcounter{enumi}{4}
\item For each pair $\wg$, generate modifiers $M^w_\vg$  for all $v \in
  N(w)$:
  \begin{itemize}
  \item If $\wg \not \in \T$ (i.e., $A_\wg = 0$) then $M^w_\vg = 1$
    for all $v \in N(w)$.
  \item Else, if $\wg \in \T$, then for all $v \in N(w)$,
    \begin{gather}
      M^w_\vg = \underbrace{2(1-\eta_\wg)}_{m.o. r.v.}
      \cdot \underbrace{r \one{v \in S\wg}}_{m.o. r.v.} \label{eq:jo-13}
    \end{gather}
    where $S\wg \sse N(w)$ is a random color class from an $r$-coloring
    of $N(w)$, this randomness is independent of all other $(w,
    \gamma')$, and $(w', \cdot)$.
  \end{itemize}
\item For each pair $\vg$, collect modifiers $M^w_\vg$ from its
  neighbors $w \sim v$, and define $M_\vg := \sprod_{w \sim v}
  M^w_\vg$. Here $\sprod$ is the stopped product (as in
  \S\ref{sec:stopped}) w.r.t.\ threshold $\phat/p\vg$.\footnote{The
    stopped product is with respect to a sequence, so let us assume a
    total order on the vertices, and the variables $\{M^w_\vg\}_{w \in
      N(v)}$ are considered in this order.} Finally set the probability
  values for the next stage to be
  \begin{gather}
    p'\vg := p\vg \cdot M_\vg\label{eq:jo-prob-update}
  \end{gather}
  Recalling the $\khat$ function from~(\ref{eq:jo-khat}), observe that
  $\khat := \max_{w,v,\g} \khat(M^w_\vg) = 2r$; hence the stopped
  product ensures $p'\vg < \khat \cdot \phat = 2r \phat$; define $\pstar
  = \khat\phat$. 
\end{enumerate}

We've now finished defining the new probabilities $p^{t+1}\vg := p'\vg$
and the new graph $G^{t+1} = G'$, and proceed to the next stage. This is
done for some $T= \Theta(\frac1\theta \ln \Delta)$ many stages, after
which we claim that the degree of $G^T$ becomes $\ll s$, and a
na\"{\i}ve coloring suffices to color the rest of the vertices.

The run time: the most expensive step is to find an $r$-coloring of the
neighborhood of the vertices. For each vertex this can be done in time 
$O(r\cdot \Delta \cdot 2^\Delta)$ using $O(2^\Delta)$ space 
(see~\cite{BHKM09}). That is followed by 
$T = O(\frac{\ln \Delta}{\theta})$ rounds of partial colorings, 
each taking $\poly(n\Delta)$ time. Hence
the total runtime is $O(nr\cdot 2^\Delta) + \poly(n)$.

\section{The Proof of Theorem~\ref{thm:jo-local}: I. The Setup}
\label{sec:jo-analysis}

The analysis of this coloring algorithm is similar in spirit (but more
technical than) Johansson's previous result for coloring triangle-free
graphs of maximum degree $\Delta$. Although that result also appears
only as an unpublished manuscript~\cite{Joh-k3}, a lucid presentation
appears in the book by Molloy and Reed~\cite{MR02}.

The idea is clever, but also natural in hindsight: as the process goes
on and probability values $p\vg$ increase, we want to show that for each
surviving vertex, its degree goes down rapidly, whereas it has many
colors still remaining. Showing the former, that many vertices in each
neighborhood are colored at each step, proceeds by showing that for each
$v$, there is not too much positive correlation between the colors of
its neighbors, and they behave somewhat independently. To show the
latter, that many colors remain valid for each vertex $v$, we show that
the entropy of the $p(v, \cdot)$ ``probability distribution'' remains
high. (The quotes are because we only have $\sum_\g p\vg \approx 1$, and
not equal to one, but this approximate equality suffices.) And
lower-bounding the entropy, as in~\cite{Joh-k3,MR02}, relies on
upper-bounding the ``energy'' of the edges which captures the positive
correlation between the colors of its endpoints, and is defined as
follows:

\begin{definition}[Energy]
  \label{def:energy}
  For an edge $uv$ the \emph{energy} with respect to the $p$ values is
  \[ \nrgy{uv}{p} := \sum_\g p\ug p\vg. \] For a non-edge $uv$, define
  $\nrgy{uv}{p} := 0$. For a graph $G$, the energy of a vertex $u$ is
  $\nrg_G(u,p) := \sum_{v \in G} \nrgy{uv}{p}$; when the graph is clear
  from context we drop the subscript and use just $\nrg(u,p)$.  
\end{definition}

\subsection{The Parameters}
\label{sec:param}

For ease of reference, we present the parameters used in the proof
here. Some of these have already been used in the algorithm description,
the others will be introduced in due course. 
\ifstoc
\begin{alignat*}{3}
  \e &:= 1/100 & \theta &:= \Delta^{-1/4 + 2\e} \\
  s &:= |L| := \frac{\Delta}{K} & \phat &:= \Delta^{-3/4 - 5\e} \\
  \khat &:= \max_{w,v,\gamma} \khat(M^w_\vg) = 2r & a &:= 1/2 - 3\e\\
  \pstar &:= \khat \phat = 2 r \phat  &  \cee &:= 0 \\
  T &:= \frac{2\e}{a \cdot \theta} \ln \Delta & \bee &:= a-\cee \\
    C &:=  \max \big( 1, \ln(\khat(M^w_\vg)) \big) = \ln(2r) &\quad K &:= \frac{(b/4) \e \ln \Delta}{C}\\
\end{alignat*}
\else
\begin{alignat*}{3}
  \e &:= 1/100 & \theta &:= \Delta^{-1/4 + 2\e} & s &:= |L| := \frac{\Delta}{K} \\
    \khat &:= \max_{w,v,\gamma} \khat(M^w_\vg) = 2r &\qquad \qquad \phat &:= \Delta^{-3/4 - 5\e} & \qquad\qquad \pstar &:= \khat \phat = 2 r \phat \\
  a &:= 1/2 - 3\e &  \cee &:= 0 & \bee &:= a-\cee \\
    C &:=  \max \big( 1, \ln(\khat(M^w_\vg)) \big) = \ln(2r) & T &:= \frac{2\e}{a \cdot \theta} \ln \Delta &\qquad\qquad K &:= \frac{(b/4) \e \ln \Delta}{C}\\
\end{alignat*}
\fi

We will assume that $r \leq \Delta^\e = \Delta^{1/100}$, else the
desired coloring number of $O(\frac{\Delta}{\ln \Delta} \ln r)$ will
just be $O(\Delta)$, which is trivial to achieve. Finally, we will
assume that $\Delta$ is large enough whenever necessary.  In particular,
\begin{gather}
    \ln \Delta \geq 10000  \label{eq:jo-delta-bound}
\end{gather}
suffices for the rest of the analysis, however no attempt has been made to
optimize any constants.

\subsection{Initial Values}
\label{sec:initial}

At the beginning, the probability values are $p^0\vg = 1/s$ for all
$\vg$, and the degrees are at most $\Delta$. Hence
\ifstoc
  \begin{align*}
    p^0(\P_v) &= 1 \tag{P0} \label{inv:p0} \\
    d(v,G^t) &\leq \Delta \tag{D0} \label{inv:d0} \\
    \nrg_{G_0}(v,p^0) &= \sum_{\gamma, u \sim v} p^0\vg p^0\ug \leq \Delta \cdot s
    \cdot 1/s^2 = K \tag{E0} \label{inv:e0} \\
    h(v,p^0) &= \sum_\gamma p^0\vg \ln 1/p^0\vg = \ln s \notag \\
    &= \ln \Delta - \ln K \tag{H0} \label{inv:h0}
  \end{align*}
\else
  \begin{align*}
    p^0(\P_v) &= 1 \tag{P0} \label{inv:p0} \\
    d(v,G^t) &\leq \Delta \tag{D0} \label{inv:d0} \\
    \nrg_{G_0}(v,p^0) &= \sum_{\gamma, u \sim v} p^0\vg p^0\ug \leq \Delta \cdot s
    \cdot 1/s^2 = K \tag{E0} \label{inv:e0} \\
    h(v,p^0) &= \sum_\gamma p^0\vg \ln 1/p^0\vg = \ln s = \ln \Delta - \ln K \tag{H0} \label{inv:h0}
  \end{align*}
\fi

\subsection{The Invariants}
\label{sec:invariants}

The proof is by induction over the stages.  We maintain the invariants
that for all $t \leq T$, for all $v \in G^t$, the following hold true:
\begin{align*}
  p^t(\P_v) &\in 1 \pm t \lambda_P \sse 1 \pm
    \sqrt{\theta} \sse 1 \pm \e \sse [1/2, 2] \tag{(InvP)} \label{inv:p} \\
  d(v,G^t) &\leq \Delta e^{-a\theta t} + t \lambda_D  \tag{(InvD)} \label{inv:d} \\
  \ts \nrg_{G_t}(v,p^t) &\leq K e^{-\bee\theta t} + \frac{\lambda_E}{\bee\theta}
    \leq 2K \tag{(InvE)} \label{inv:e} \\
  h(v,p^t) &\geq h(v, p^0) - \theta C \sum_{t' \leq t} \nrg_{G_{t'}}(v,\pa^{t'})
    - t \lambda_H  \notag \\
    &\geq (1 - \e) \ln \Delta \tag{(InvH)} \label{inv:H}
\end{align*}
where 
\begin{gather}
  \lambda_P = O(\sqrt{\pstar \ln \Delta}) \qquad
\lambda_D = O(K\theta^2\Delta) \notag \\
\lambda_E = O(\theta \Delta^{-2\e} \ln \Delta) \qquad
\lambda_H = O(\sqrt{\pstar \ln^3 \Delta})
\end{gather}
In the following, we will show that the tightest bounds hold for each
$t$. The weaker bounds given above are just implications useful in our
proofs, and in all cases these follow by algebra. Observe the initial
values from \S\ref{sec:initial} satisfy these invariants.

\section{The Proof of Theorem~\ref{thm:jo-local}: II. The Inductive Step}
\label{sec:induc}

We assume the invariants hold for all times upto and including time $t$,
and now want to show these are satisfied at the end of stage $t+1$. As
usual, we use $p$ to denote $p^t$, and $p'$ to denote $p^{t+1}$. The
plan is to show that for each fixed vertex $v$, each of the invariants
hold with high probability. Since we cannot take a union bound without
losing terms dependent on $n$, we show that the failure events depend
only on a small number of other failure events, whereupon we can apply
the Lov\'asz Local Lemma to complete the argument. 

In the following arguments, we assume that $d(v) \geq \frac{\Delta}{\ln
  \Delta}$ for all $v \in G_t$. Indeed, at the beginning of any stage we
can repeatedly remove vertices with degree less than $\frac{\Delta}{\ln
  \Delta}$, and having found a coloring for the rest of the vertices, we
can color these removed vertices greedily at the end.

\subsection{The Total Probability}
\label{sec:probability}

By construction, the probability values $p^t\vg$ form a martingale, and
hence it is not surpising that their sum remains concentrated around
$1$. Here is the formal proof.

\begin{lemma}
  \label{lem:prob}
  For vertex $v$, if $p(\P_v) \in 1 \pm \e$, then $p'(\P_v) = p(\P_v)
  \pm O(\sqrt{\pstar \ln \Delta})$ \whpD.
\end{lemma}

\begin{proof}
  By construction of the modifiers as mean-one r.v.s, $\Exp[p'\vg] =
  p\vg \cdot \Exp[M_\vg] = p\vg$ for all $\vg$. Moreover, $p'(\P_v)$ is
  the sum of independent $\pstar$-bounded random variables, with
  $\Exp[p'(\P_v)] \leq (1 + \e) \leq 2$, Theorem~\ref{thm:chernoff}
  implies that the deviation $\abs{p'(\P_v) - p(\P_v)}$ is at most
  $\lambda_P := O(\sqrt{\pstar \ln \Delta} + \pstar \ln \Delta) =
  O(\sqrt{\pstar \ln \Delta})$ \whpD.
\end{proof}

\subsubsection{The Lov\'asz Local Lemma Argument} 
\label{sec:lll-arg}

To get the property of Lemma~\ref{lem:prob} for all vertices $v$
simultaneously requires the LLL, since we cannot take a union bound over
all the $n$ vertices. For this, define the bad event $\B^p_{v} =
\{\omega: p'(\P_v) > p(\P_v) + \lambda_P \}$. Note that this event
depends only on a subset of the variables $\{ A_\ug, \eta_\ug, M_\ug\}$
for $u \in \{v\} \cup N(v)$ and $\g \in L$, and all these $A, \eta, M$
variables are independent. Hence, $\B^p_{v}$ and $\B^p_w$ are clearly
independent if $N(u) \cap N(v) = \emptyset$, and the dependency graph
has degree at most $\Delta^2 s \leq \Delta^3$. Using the Moser-Tardos
framework~\cite{MT10} we can get a coloring where none of these bad events
happen.

For each of the the other invariants, we use a similar approach using
the LLL: we define ``local'' bad events --- i.e., the bad event at a
vertex will depend only on r.v.s for vertices at some constant distance
from it --- and hence the degree of the dependency graph over the bad
events will be bounded by $\poly(\Delta)$. Moreover, the probability of
each bad event will be at most $1/\poly(\Delta)$. Using the Moser-Tardos
framework will allow us to find an outcome that will avoid all the bad
events at all vertices simultaneously. Since the arguments will be very
similar, we will henceforth just explain what the bad events are, and
omit the details. 

\subsection{The Degree}
\label{sec:degree}

\begin{lemma}
  \label{lem:degree}
  For vertex $v$, \whpD the new degree is
  \[d'(v) \leq (1 - a \theta)\, d(v) \pm O(K\theta^2 \Delta). \]
\end{lemma}

It is more convenient to study $X_u$, the indicator of whether $u$ was
colored in this round; let $X = \sum_{u \sim v} X_u$. Then $d'(v) = d(v)
- X$. We will first show that $\Exp[X]$ is large, and then that $X$ is
concentrated around its mean. Let
\begin{gather}
  Y_u := 1 - \prod_{\gamma} ( 1 - A_\ug \eta_\ug)
\end{gather}
indicate whether
$u$ was tentatively assigned at least one color that was not dropped due
to the $\eta$ coin flip; clearly $X_u \leq Y_u$. Moreover, let
\begin{gather}
  Y'_u := \sum_\gamma \sum_{w \sim u} A_\ug A_\wg.
\end{gather}
It is easy to see that $Y_u - Y'_u \leq X_u$.

\begin{claim}
  \label{clm:xu-yu}
  $\Pr[ u \text{ is colored}\,] = \Exp[X_u] \geq \left(\frac12 - \e 
  \right) \theta \geq (a + \e) \theta$.
\end{claim}

\begin{proof} By inclusion-exclusion,
  \ifstoc
    \begin{align*} 
      \Exp[Y_u] \geq \sum_\g \frac{\theta}{2}\pc\ug - \sum_{\g, \g'}
      \frac{\theta^2}{4}\pc\ug \pc(u, \g') \notag \\
      = \frac{\theta}{2}\pc(\P_u)\bigg(
      1 - \frac{\theta}{2} \pc(\P_u) \bigg). 
    \end{align*} 
  \else
    \begin{align*} 
      \Exp[Y_u] \geq \sum_\g \frac{\theta}{2}\pc\ug - \sum_{\g, \g'}
      \frac{\theta^2}{4}\pc\ug \pc(u, \g') = \frac{\theta}{2}\pc(\P_u)\bigg(
      1 - \frac{\theta}{2} \pc(\P_u) \bigg). 
    \end{align*} 
  \fi
  Using~\ref{inv:p},~\ref{inv:e} and~\ref{inv:H} in Lemma~\ref{lem:p-to-pa}, we know
  that $\pc(\P_u) \geq 1 - 6(\sqrt{\theta} + \theta) -2\e \geq 1 -
  12\sqrt{\theta} - 2\e$. 
  Also, by~\ref{inv:p}, $\pc(\P_u) \leq p(\P_u) \leq 2$, so by algebra
  we get $\Exp[Y_u] \geq \theta (1/2 -\e) - O(\theta^{3/2})$.

  Moreover,
  \[ \Exp[Y'_u] = \theta^2 \sum_\gamma \sum_{w \sim u} \pc\ug \pc\wg =
  \theta^2 \nrg(u, \pc). \] Using that $\nrg(u, \pc) \leq \nrg(u, p) \leq 2K$ for all $u$
  (from~\ref{inv:e}), we infer $\Exp[Y'_u] \leq 2K\theta^2$. Since
  $K\theta^2 = O(\theta^{3/2})$, 
  \ifstoc
    \begin{align*}
      \textstyle \Pr[u \text{ colored} ] = \Exp[X_u] \geq \Exp[Y_u - Y_u'] \\
      \geq \theta(1/2-\e) - O(\theta^{3/2}) \geq
      \left( \frac12 - 2\e \right) \theta 
    \end{align*} 
  \else
    \begin{align*}
      \textstyle \Pr[u \text{ colored} ] = \Exp[X_u] \geq \Exp[Y_u - Y_u'] 
      \geq \theta(1/2-\e) - O(\theta^{3/2}) \geq
      \left( \frac12 - 2\e \right) \theta 
    \end{align*} 
  \fi
  The final inequality holds for large enough $\Delta$~(\ref{eq:jo-delta-bound}), and the
  claim follows using the definition of $a$.
\end{proof}

\begin{corollary}
  \label{cor:exp-degree}
  $\Exp[d'(v)] \leq (1 - a \theta)\cdot d(v)$
\end{corollary}

\begin{proof}
  Suppose $X := \sum_{u \sim v} X_u$ is the expected number of neighbors
  of $v$ that get colored; by Claim~\ref{clm:xu-yu}, we get that \[
  \Exp[X] \geq \sum_{u \sim v} \Exp[X_u] \geq d(v) \cdot a\theta. \] Finally, observing that $d'(v) = d(v) - X$
  completes the proof.
\end{proof}

\subsubsection{The Concentration Bound for Degrees}
\label{sec:degree-conc}

We want to show that $X - \Exp[X]$ is small \whpD. Let $Y := \sum_{u
  \sim v} Y_u$, and $Y' := \sum_{u \sim v} Y'_u$. For the upper tail,
observe that as $Y - Y' \leq X \leq Y$, we have
\begin{gather}
  X - \Exp[X] \leq Y-\Exp[Y-Y'] = (Y - \Exp[Y]) + \Exp[Y'].\label{eq:jo-4}
\end{gather}
But $Y$ is the sum of independent $\{0,1\}$-valued r.v.s $\{Y_u\}_{u
  \sim v}$, and by inclusion-exclusion again each $\Exp[Y_u] \leq
\sum_\g \frac\theta2 p\ug$. Thus
\[ \Exp[Y] \leq \theta/2 \sum_{u \sim v} \sum_\g \pc\ug \leq O(\theta
d(v)) = O(\theta\Delta). \] 
By the tail bound Theorem~\ref{thm:chernoff}, setting $\lambda_D^{(1)} := O(\sqrt{\theta\Delta
  \ln \Delta} + \ln \Delta)$ suffices to give
$\Pr[\abs{Y - \Exp[Y]} \leq \lambda_D^{(1)}] \leq 1/\poly(\Delta)$. Plugging
this into~(\ref{eq:jo-4}) and using the bound of $\Exp[Y'] =
\Exp[\sum_{u\sim v} Y_u'] \leq 2K\theta^2 \Delta =: \lambda^{(2)}_D$ gives
us that the total deviation of $X$ above its mean is at most $\lambda_D^{(1)}
+ \lambda^{(2)}_D = O(K\theta^2 \Delta)$ \whpD.

For the lower tail, observe that
\ifstoc
  \begin{gather}
    \Exp[X] - X \leq \Exp[Y]-(Y-Y') \notag \\
    = (\Exp[Y] - Y) + \Exp[Y'] + (Y' - \Exp[Y']).\label{eq:jo-5}
  \end{gather}
\else 
  \begin{gather}
    \Exp[X] - X \leq \Exp[Y]-(Y-Y') = (\Exp[Y] - Y) + \Exp[Y'] + (Y' - \Exp[Y']).\label{eq:jo-5}
  \end{gather}
\fi
Since $\Exp[Y] - Y \leq \lambda_D^{(1)}$ and
$\Exp[Y'] \leq \lambda^{(2)}_D$ by the preceding argument, so we focus on
bounding the upper tail of $Y'$.  For this we use the concentration
inequality for polynomials from Theorem~\ref{thm:ss-poly}. The
parameters are:

\ifstoc
  \begin{align}
    \mu_0 &= \Exp[Y'] \leq d(v)\cdot 2K\theta^2 \leq 2K\theta^2 \Delta\\ 
    \mu_1 &\leq 2\max_\vg \sum_{u \sim
      v} \Exp[A_\ug] \notag \\
   &= 2\max_\vg \sum_{u \sim v} \theta \pc\ug \leq 2\theta
    \phat \Delta = O(\Delta^{-3\e}) \leq 1 \\
    \mu_2 &= 2.
  \end{align}
\else
  \begin{align}
    \mu_0 &= \Exp[Y'] \leq d(v)\cdot 2K\theta^2 \leq 2K\theta^2 \Delta\\ 
    \mu_1 &\leq 2\max_\vg \sum_{u \sim
      v} \Exp[A_\ug] 
   = 2\max_\vg \sum_{u \sim v} \theta \pc\ug \leq 2\theta
    \phat \Delta = O(\Delta^{-3\e}) \leq 1 \\
    \mu_2 &= 2.
  \end{align}
\fi 
Plugging this into Corollary~\ref{cor:ss-poly-2} of the aforementioned
concentration inequality, we get that \whpD the deviation $\abs{Y' -
  \Exp[Y']}$ is $O(\sqrt{2K\theta^2 \Delta \ln \Delta} + \ln^2 \Delta)
=: \lambda^{(3)}_D$. Substituting into~(\ref{eq:jo-5}), we have that \whpD,
\begin{gather}
  \Exp[X] - X \leq \lambda^{(1)}_D + \lambda^{(2)}_D + \lambda^{(3)}_D = O(K
  \theta^2 \Delta) =: \lambda_D
\end{gather}
This proves Lemma~\ref{lem:degree}.

Finally, the LLL part: here the bad event $\B^d_v = \{ \omega: d'(v) > (1
- a\theta)d(v) + \lambda_D) \}$, and this depends on the $A, \eta, M$ r.v.s for
vertices at distance at most $2$ from $v$ (since it depends on whether
$u \in N(v)$ survive, which depend on their neighbors). Hence the
dependency is at most $\Delta^4 \times s$.

\subsection{The Entropy}
\label{sec:entropy}

Now to show invariant~\ref{inv:H}, that the entropy of $\{p'\vg\}_\g$
remains high, where the entropy is defined as
\begin{gather}
  h(v, p') := - \sum_\g p'\vg \ln p'\vg
\end{gather}

\begin{lemma}
  \label{lem:entropy}
  $\Exp[h(v, p')] \geq h(v,p) - C \theta\, \nrg(v, \pa)$.
\end{lemma}

\begin{proof}
  Recall that $p'\vg := p\vg \cdot M_\vg$ from~(\ref{eq:jo-prob-update}),
  where $M_\vg$ is a \morv.
  \begin{align}
    h(v, p') &= - \sum_\g p'\vg \ln p'\vg \notag \\
    &= - \sum_\g p\vg M_\vg \ln p\vg - \sum_\g p\vg M_\vg \ln
    M_\vg \notag
  \end{align}
  Hence, taking expectations,
  \ifstoc
    \begin{align*}
      \Exp[h(v, p')] &= - \sum_\g \Exp[M_\vg]\, p\vg \ln p\vg \notag \\
        & \qquad - \sum_\g p\vg
      \, \Exp[M_\vg \ln M_\vg] \notag \\
      &= h(v,p) - \sum_\g p\vg \, \Exp[M_\vg \ln M_\vg] \notag
    \end{align*}
  \else
    \begin{align*}
      \Exp[h(v, p')] &= - \sum_\g \Exp[M_\vg]\, p\vg \ln p\vg - \sum_\g
      p\vg
      \, \Exp[M_\vg \ln M_\vg] \notag \\
      &= h(v,p) - \sum_\g p\vg \, \Exp[M_\vg \ln M_\vg] \notag
    \end{align*}
  \fi
  In fact, if the probability $p\vg$ is greater than $\phat$ for some
  $\g$ (i.e., if $\pa\vg > 0$), then the definition of the stopped
  product implies that $M_\vg \equiv 1$. Hence, we get the stronger
  claim that
  \begin{align}
    \Exp[h(v, p')] &= h(v,p) - \sum_\g \pa\vg \, \Exp[M_\vg \ln
    M_\vg] \label{eq:jo-1}
  \end{align}
  Now, recall that for an r.v.\ $X$, we defined $\kappa(X) = \Exp[X \ln
  X]$ in \S\ref{sec:kappa}.

  \begin{claim}
    \label{cl:kappa}
    $\kappa(M_\vg) = \Exp[M_\vg \ln M_\vg] \leq \theta\, C\, \sum_{w \sim v} \pc\wg$.
  \end{claim}

  \begin{proof}
    Observe that $M_\vg$ is a stopped product of a bunch of \morvs
    $M^w_\vg$ of neighbors $w \sim v$, each of which is either $1$ (if
    $\wg \not\in \T$) or itself a product of two \morvs as
    in~(\ref{eq:jo-13}).
    Using properties of the $\kappa(\cdot)$ function from
    \S\ref{sec:kappa}, we get
    \begin{gather}
      \kappa(M_\vg) \leq \sum_{w \sim v} (\theta \pc\wg)\cdot
      \max \kappa( M^w_{v,\gamma} )
    \end{gather}
    Using the definition of $C$ gives us the claim.
  \end{proof}

  Substituting Claim~\ref{cl:kappa} into~(\ref{eq:jo-1}), 
  \ifstoc
    \begin{align}
      \Exp[h(v, p')] &\geq h(v,p) - \theta\, C\, \sum_\g\sum_{w \sim v} \pa\vg \pc\wg \notag \\ 
      &\geq h(v,p) - \theta\, C\, \nrg(v,\pa).
    \end{align}
  \else
    \begin{align}
      \Exp[h(v, p')] &\geq h(v,p) - \theta\, C\, \sum_\g\sum_{w \sim v}
      \pa\vg \pc\wg \geq h(v,p) - \theta\, C\, \nrg(v,\pa).
    \end{align}
  \fi
  This proves the lemma.
\end{proof}

\begin{lemma}
  \label{lem:entropy-whp}
  The deviation $\abs{h(v,p') - \Exp[h(v,p')]}$ is at most $\sqrt{\pstar
      \ln^3 \Delta}$ \whpD.
\end{lemma}

\begin{proof}
  By definition, the entropy $h(v,p')$ is a sum of independent r.v.s
  $p'\vg \ln \frac1{p'\vg}$; since $p'\vg \in [1/s, \pstar]$, these
  r.v.s are $[0,m]$-bounded where $m := (\pstar \ln s) \leq \pstar \ln
  \Delta$. Moreover, the mean $\mu$ satisfies $\mu \leq \ln \Delta$
  because the entropy can be at most $\ln s \leq \ln \Delta$. Hence by
  the tail bound from Theorem~\ref{thm:chernoff}, the deviation from the mean is at most
  $\lambda_H = O(\sqrt{\mu m \ln \Delta} + m \ln \Delta) =
  \sqrt{\pstar \ln^3 \Delta}$ \whpD.
\end{proof}

For the LLL application, since the entropy depends only on the $p'\vg$
values, the bad event for $v$ is dependent only on the random choices of
$\{v\} \cup N(v)$, and hence easily handled for the usual reasons.

\subsection{The Energy}
\label{sec:energy}

The calculations above show that the decrease in entropy at vertex $v$
depends on the energy of edges incident to $v$, so it remains to show
that this energy is small (and in fact, decreases over the course of the
algorithm). This is technically the most interesting part of the
analysis. We're interested in
\begin{gather}
  \nrg_{G}(v, \pa) = \sum_{u \sim_{G} v} \underbrace{\nrg(uv,
  \pa)}_{\sum_\g \pa \ug \pa\vg},
\end{gather}
and want to show that this energy drops by a factor of $\approx (1 -
\theta/2)$ each time. Formally, we prove the following:
\begin{lemma}
  \label{lem:energy-whp}
  For a vertex $v$, \whpD the energy \[ \nrg_{G'}(v, \pa') \leq (1 -
      \bee \theta) \cdot \nrg_G(v, \pa) + O(\theta \Delta^{-2\e}
  \ln \Delta). \]
\end{lemma}

As for the other invariants, we first bound the expectation and then show a
large-deviations bound. The expectation calculation itself proceeds via
two claims --- the first claim quantifies the change in energy due to
considering the new probability distribution $\pa'$ instead of $\pa$
(keeping the graph $G$ fixed), and the second captures the change due
to considering graph $G'$ instead of $G$ (but now keeping the
distribution $p'$ fixed). Remember that the energy of non-edges
is zero by definition.
\begin{claim}
  \label{clm:energy1}
  For any $v$, $\sum_{u \in G} \Exp[\nrgy{uv}{\pa'}] \leq \sum_{u \in G} 
    \nrgy{uv}{\pa}$.
\end{claim}

\begin{claim}
  \label{clm:energy2}
  For any $v$, $\Exp\left[\sum_{u \in G'} \nrgy{uv}{\pa'}\right] \leq 
  (1 - a\theta) \sum_{u \in G} \Exp[\nrgy{uv}{\pa'}]$.
\end{claim}
In the latter claim, observe that $G'$ is a random variable itself, and
hence we cannot just push the expectation inside the sum.  Before we
prove Claim~\ref{clm:energy1}, we will use the following observation
which allows us to consider the unstopped product instead of the stopped
product.
\begin{fact}
    \label{fct:stopped}
    For any $u,v\in G$, and $\gamma \in L$, 
    $\pa'\ug \pa'\vg \leq \pa \ug \pa \vg \prod_{w\sim u} M^w_\ug \prod_{w \sim u} M^w_\ug$
\end{fact}

\begin{proof}
  Indeed, the products are stopped only when the terms multiplied give
  us value at least $\phat$, but then corresponding $\pa'$ values get
  zeroed out.
\end{proof}
    
\begin{proofof}{Claim~\ref{clm:energy1}} 
  Using Fact~\ref{fct:stopped}, we can replace the stopped product by
  the usual ones:
  \ifstoc
    \begin{align*}
      \Exp[\pa'\ug \pa'\vg]
      &\leq \pa\ug \pa\vg \Exp\left[ \prod_{w \sim u} M^w_\ug \prod_{x \sim v}
        M^x_\vg\right]  \\
      \leq p\ug p\vg &\Exp\left[
        \prod_{w \sim u: w \not\sim v} M^w_\ug\right] \cdot \Exp\left[ \prod_{x
          \sim v: x \not\sim u} M^x_\vg \right] \notag \\
        &\cdot \Exp\left[ \prod_{w: uvw
          \in \triangle} M^w_\ug M^w_\vg \right] \label{eqn:energy1}
    \end{align*}
  \else
    \begin{align*}
      \Exp[\pa'\ug \pa'\vg]
      &\leq \pa\ug \pa\vg \Exp\left[ \prod_{w \sim u} M^w_\ug \prod_{x \sim v}
        M^x_\vg\right]  \\
      &\leq p\ug p\vg \Exp\left[
        \prod_{w \sim u: w \not\sim v} M^w_\ug\right] \Exp\left[ \prod_{x
          \sim v: x \not\sim u} M^x_\vg \right] \Exp\left[ \prod_{w: uvw
          \in \triangle} M^w_\ug M^w_\vg \right] \label{eqn:energy1}
    \end{align*}
  \fi
  But the first two expectations equal to one, since each of these is a
  product of independent \morvs. For the third one, with probability $(1
  - \theta \pc\wg)$ the pair $\wg \not\in \T$ and we get $1$, else with
  probability $\theta \pc\wg$ the pair $\wg \in \T$ and the random color
  class modifier gives zero. In all cases the final expectation would be
  at most one.
\end{proofof}

\begin{proofof}{Claim~\ref{clm:energy2}} 
  In this claim, it suffices to bound,
  for each color $\g$, $\Exp\left[ \pa'\ug \pa'\vg \one{u \in G'}
  \right]$. The crucial observation is that for any $\g$, the events
  \begin{align}
    \{ u \in G' \} &= \{ \C_u = \emptyset \} \\
    &= \{ \C_u \sse \{\ug\} \} - \{ \C_u = \{\ug\} \} \\
    &= \{ \C \cap (\{u\} \times (L \setminus \{\gamma\})) = \emptyset \}
    - \{ \C_u = \{\ug\} \} \label{eq:jo-3}
  \end{align}
  However, when $\C_u = \{\ug\}$, then $A_\ug = 1$ and $\eta\ug = 1$,
  which means that $\pa'\vg = 0$ due to the modifier $M^u_\vg$, and hence
  $\Exp[ \pa'\ug \pa'\vg \one{\C_u = \{\ug\}}] = 0$. Moreover, the former
  event says that $u$ did not get any color from the set $L \setminus
  \gamma$, which is independent of all decisions for the
  color~$\g$. Hence
  \begin{gather}
    \Exp\left[ \pa'\ug \pa'\vg \one{u \in G'} \right] = \Exp[ \pa'\ug \pa'\vg ]
    \cdot \Pr[ \C \cap (\{u\} \times (L \setminus \{\gamma\})) =
    \emptyset ] \label{eq:jo-6}
  \end{gather}
  Using~(\ref{eq:jo-3}) again, we know that
  \begin{align}
    \Pr[ \C \cap (\{u\} \times (L \setminus \{\gamma\})) = \emptyset ]
    &=
    \Pr[ u \in G' ] + \Pr[ \C_u = \{\ug\} ] \\
    &\leq  \Pr[ u \in G' ] + \Pr[ A_\ug = 1 ] \label{eq:jo-7} \\
    &\leq (1 - (a + \e) \theta) + \theta\phat, \label{eq:jo-9}
  \end{align}
  where the first expression is from Claim~\ref{clm:xu-yu}
  and the second one from $\pc\ug \leq \phat$. But
  $\theta\phat \leq \e\theta$ for large $\Delta$. Thus the claim is proved.
\end{proofof}

Combining Claims~\ref{clm:energy1} and~\ref{clm:energy2},  and using $a =
 b$, 
\begin{gather}
  \Exp[ \nrg_{G'}(v,\pa')] = \sum_{u\in G'} \Exp[\nrg(uv,\pa')] \leq
  \left(1-\bee\theta \right) \sum_{u\in G} \nrg(uv,\pa) = \left(1-\bee\theta
  \right) \nrg_G(v,\pa). \label{eq:jo-14}
\end{gather}
We next show that the r.v.\ $\nrg_{G'}(v,\pa')$ is concentrated around its
mean.

\subsubsection{The Concentration Bound for Energy}
\label{sec:energy-conc}

Fix some $v \in G$. We want to show that $\sum_{u \in G} \nrgy{uv}{\pa'}
\one{u \in G'}$ is concentrated around its mean. The idea is simple:
denoting $Q_u = \nrgy{uv}{p'}$ and $R_u = \one{u \in G'}$, and letting
$q_u = \Exp[Q_u], r_u = \Exp[R_u]$ being their expectations, the triangle inequality gives us that
\[ |\sum_u Q_uR_u - \Exp[Q_uR_u]| \leq |\sum_u (Q_u - q_u)\,R_u| + |\sum_u
q_u(R_u - r_u)| + |\sum_u q_ur_u - \Exp[Q_uR_u]| \]
The following three claims now bound the three expressions on the right.

\begin{claim}
  \label{clm:energy1whp}
  $\abs{\sum_{u \in G} \nrgy{uv}{\pa'} \one{u \in G'} - \sum_{u \in G}
    \Exp[\nrgy{uv}{\pa'}] \one{u \in G'}} \leq O(\theta \Delta^{-7\e} \ln
  \Delta)$ \whpD.
\end{claim}

\begin{proof}
  The left hand side is at most $\sum_{u \in N_G(v)} \abs{
    \nrgy{uv}{\pa'} - \Exp[\nrgy{uv}{\pa'}]}$. Each $\nrgy{uv}{\pa'}$ is a sum of
  independent $\phat^2$-bounded r.v.s, and hence deviates from its mean by at most
  $O(\sqrt{\Exp[\nrgy{uv}{\pa'}] \phat^2 \ln \Delta} + \phat^2 \ln
  \Delta)$ \whpD by the tail bound in Theorem~\ref{thm:chernoff}. Summing this over all $u \in N_G(v)$ (and taking a
  union bound over these $\Delta$ events), we get the total deviation
  \whpD is
  \begin{align*}
    &\sum_{u \in N_G(v)} O(\phat \sqrt{\Exp[\nrgy{uv}{\pa'}] \ln \Delta} +
    \phat^2 \ln \Delta) \\
    &\leq O\bigg(\phat \sqrt{\ln \Delta}\cdot \sqrt{\Delta} \, \sqrt{\Exp\bigg[ \sum_{u \in
      N_G(v)} \nrgy{uv}{p'} \bigg]}\bigg) + O(\phat^2 \Delta \ln \Delta) \\
    &\leq O(\phat \sqrt{\Delta K \ln \Delta} + \phat^2 \Delta \ln
    \Delta).
  \end{align*}
  The first inequality uses Cauchy-Schwarz; the next one uses the
  expectation bound from Claim~\ref{clm:energy1} and
  invariant~\ref{inv:e}. The dominant term is $O(\phat \sqrt{\Delta K
    \ln \Delta}) = O(\theta \Delta^{-7\e} \ln \Delta)$, which completes
  the claim.
\end{proof}

\begin{claim}
  \label{clm:energy2whp}  
  $\abs{\sum_{u \in G} \Exp[\nrgy{uv}{\pa'}] \one{u \in G'} -
    \sum_{u \in G} \Exp[\nrgy{uv}{\pa'}] \Exp[\one{u \in G'}] } \leq 
  O(\theta \Delta^{-2\e} \ln \Delta)$ \whpD.
\end{claim}

\begin{proof}
  By Claim~\ref{clm:energy1}, each term 
  \begin{gather}
    \Exp[\nrgy{uv}{\pa'}] \leq \nrgy{uv}{\pa} = \sum_\g \pa\ug \pa\vg \leq
    \phat \sum_\g \pa\vg \leq 2\phat. \label{eq:jo-8}
  \end{gather}
  Using $c_u :=
  \Exp[\nrgy{uv}{\pa'}]$ and observing that $c_u = 0$ for all $u \not\in
  N_G(v)$, we want to bound the deviation
  \[ \ts \abs{ \sum_{u \in N_G(v)} c_u \left( \one{u \in G'} - \Pr[ u
      \in G'] \right) }. \] The argument now follows that of
  \S\ref{sec:degree-conc}, with the only difference that the
  variables are in $[0,2\phat]$ rather than $\{0,1\}$, which merely
  multiplies the deviation from the mean by a factor of $2\phat$. Hence,
  \whpD, we have
  \[ \ts \abs{ \sum_{u \in N_G(v)} c_u \one{u \in G'} - \sum_{u \in
      N_G(v)} c_u \Pr[ u \in G']} \leq O(\theta^2 K \phat \Delta) =
  O(\theta \Delta^{-2\e} \ln \Delta). \]
  This proves the claim. 
\end{proof}

\begin{claim}
  \label{clm:energy3whp}
  $\abs{\sum_{u \in G} \Exp[\nrgy{uv}{\pa'}] \Exp[\one{u \in G'}] - \sum_{u \in G}
    \Exp[\nrgy{uv}{\pa'} \one{u \in G'}] } \leq O(\theta \Delta^{-1/2})$ \whpD.
\end{claim}

\begin{proof}
  Using~(\ref{eq:jo-6}) and~(\ref{eq:jo-7}), summing over all colors $\g$, and
  using the fact that $\Pr[ A_\ug = 1 ] \leq \theta \phat$ for all $\g$, 
  \[ 
  \Exp[ \nrgy{uv}{\pa'}] \Exp[ \one{u \in G'} ] \leq \Exp[ \nrgy{uv}{\pa'}
  \one{u \in G'} ] \leq \Exp[ \nrgy{uv}{\pa'}] \left( \Exp[ \one{u \in G'}
    ] + \theta \phat \right). \] Rearranging, summing over all $u \in
  G$, and using that $\Exp[\nrgy{uv}{\pa'}] \leq
  2\phat$ by the calculation in~(\ref{eq:jo-8}), 
  \begin{align}
    0 \leq \sum_{u \in G}
    \Exp[\nrgy{uv}{\pa'} \one{u \in G'}] - \sum_{u \in G}
    \Exp[\nrgy{uv}{\pa'}] \Exp[\one{u \in G'}] &\leq \sum_{u \in G}
    \Exp[\nrgy{uv}{\pa'}] \theta \phat \\
    &\leq \Delta \cdot 2\theta \cdot \phat^2.
  \end{align}
  This is at most $O(\theta \Delta^{-1/2})$, which proves the claim.
\end{proof}

Putting Claims~\ref{clm:energy1whp}--\ref{clm:energy3whp} together, we
get that \whpD
\begin{align}
  \sum_{u \in G} \nrgy{uv}{\pa'} \one{u \in G'} &\leq \sum_{u \in G}
    \Exp[\nrgy{uv}{\pa'} \one{u \in G'}] + O(\theta \Delta^{-2\e} \ln \Delta) \\
  &\leq (1 - \bee \theta) \cdot \nrgy{uv}{\pa} + O(\theta
  \Delta^{-2\e} \ln \Delta) 
\end{align}
where the latter inequality follows from Claims~\ref{clm:energy1}
and~\ref{clm:energy2}.

Finally, for the LLL, the bad events in this case are again dependent
only on the random choices within distance $2$ of $v$, which means the
dependency is $O(\Delta^4 s)$.

\subsection{Behavior after \texorpdfstring{$T$}{T} rounds, and Maintaining 
        Invariants}
\label{sec:inv}

In the previous sections, we showed that if the invariants
\ref{inv:p}--\ref{inv:H} held at the beginning of a specific round, then
using the LLL we can ensure that they hold at the end of the lemma, with
some additional loss. Using this, we now show that the bounds we have
derived suffice for the invariants hold at each round in $[1..T]$. For
all these bounds we use that $\Delta$ is large enough~(\ref{eq:jo-delta-bound}).

\begin{itemize}
\item \textbf{Probability.} After each round, the probability value may
  increase by $\lambda_P$. This means after $T$ rounds,
  \[ p^t(\P_v) \leq 1 \pm T \lambda_P = 1 \pm \frac{2\e \ln(\Delta)}{a\theta} \ln
  \Delta \cdot \theta \Delta^{-1/8-4\e} \sqrt{\ln \Delta} = 1 \pm
  \sqrt{\theta}. \]

\item \textbf{Degree.} In each round, the degree falls by a
  multiplicative factor of $(1 - a\theta)$, but
  potentially increases by an additive term of $\lambda_D$. This means
  that after $t$ rounds, the degree can be (wastefully) bounded
  by
  \[ d(v, G^t) \leq (1 - a\theta)^t \cdot \Delta +
  t \lambda_D \leq e^{-a\theta t} \Delta + O(K \theta^2 \Delta t) \]
  Now for $t = T$, we get that
  \begin{gather}
    d(v, G^t) \leq \Delta e^{-2 \e \ln \Delta} + O(K \theta \Delta) \leq
    \Delta^{1-\e}. \label{eq:jo-2}
  \end{gather}

\item \textbf{Energy.} Again the energy falls by a $(1-\bee\theta)$ factor
  but potentially increases by an additive term of $\lambda_E$. This
  time we use a slightly better bound\footnote{Given a system $x_{t+1}
    \leq \alpha x_t + \beta$ for $\alpha \leq 1$, we know that $x_t
    \leq \alpha^t x_0 + \frac{\beta}{1 - \alpha}$.} of
  \begin{align}
    \nrg(v, \pa^t) &\leq (1 - \bee\theta)^t \nrg(v, \pa^0) + \frac{\lambda_E}{\bee
      \theta} \label{eq:jo-10} \\
    &\leq e^{-\bee\theta t} K + O(\Delta^{-2\e} \ln \Delta) \leq
    2K.
  \end{align}

\item \textbf{Entropy.} From Lemmas~\ref{lem:entropy}
  and~\ref{lem:entropy-whp} we get that
  \begin{align*}
    h(v, p^t) &\geq h(v, p^0) - \theta C \sum_{t' \leq t} \nrg(v,
    p^{t'}_a) - t \lambda_H\\
    &\geq (\ln \Delta - \ln K) - \theta C \sum_{t' \leq t} \left( (1 - \bee\theta)^{t'}
    \nrg(v, \pa^0) + \frac{\lambda_E}{\bee\theta} \right) - t \sqrt{\pstar \ln^3 \Delta} \\
    &\geq (\ln \Delta - \ln K) - \theta C \left( 
    \frac{\nrg(v, \pa^0)}{\bee\theta} + \frac{t \lambda_E}{\bee\theta} \right) - t \sqrt{\pstar \ln^3 \Delta} \\
    &\geq (\ln \Delta - \ln K) -
        \frac{KC}{\bee } - \frac{t C \,\theta \cdot
      O(\Delta^{-2\e} \ln \Delta)}{\bee}  - t \sqrt{\pstar \ln^3 \Delta} 
  \end{align*}
  For $t \leq T = \frac{\e}{a\theta} \ln \Delta$, we get
  \begin{align*}
    h(v, p^t) &\geq (\ln \Delta - \ln K) - \frac{KC}{\bee}- 
    \frac{2\e \ln(\Delta)}{a\cdot \bee} O( \Delta^{-2\e} \ln^2 \Delta) 
    - O(\Delta^{-1/8} \ln^{2.5} \Delta) \\
    &\geq (1 - \e)\ln \Delta.
  \end{align*}

\end{itemize}

By~(\ref{eq:jo-2}), the degree of all surviving vertices falls below
$\Delta^{1-\e}$ after $T$ rounds, whence we can color them using
$\Delta^{1-\e}$ more colors. Hence the total number of colors used is $s
+ \Delta^{1-\e} = O(\frac{\Delta}{K}) = O(\frac{\Delta \cdot C}{\ln
  \Delta})$. It remains to bound the parameter $C$: since the \morvs
$M^w_\vg$ either take on value $0$ or $2r$, we know that
$\kappa(M^w_\vg) = \ln (2r) = \ln 2 + \ln r$. This means the number of
colors used is
\[ O\left(\frac{\Delta}{\ln \Delta} \; \ln r\right) \]
This completes the proof of Theorem~\ref{thm:jo-local}.

\section{\texorpdfstring{$K_r$}{K-r}-free Graphs}
\label{sec:kr-coloring}

The above analysis was tailored for graphs where each neighborhood is
$r$-colorable. To color $K_r$-free graphs, we use different modifiers
which give a weaker guarantee of $O(\frac{\Delta}{\ln \Delta}\, (r^2 + r
\ln \ln \Delta))$ colors as claimed in Theorem~\ref{thm:jo-kr}. Since a
coloring algorithm with $s$ colors gives an independent set of size
$n/s$, this result matches Shearer's bound for independent sets for
values of $r \in O(\ln \ln \Delta)$.

In this section, we will assume that $r \leq  c' \sqrt{\log \Delta}$ for
some suitably small constant $c'$; for values of $r \geq c' \sqrt{\log
  \Delta}$, the quantity $O(r^2 + r \ln \ln \Delta)$ is $\Omega(\ln
\Delta)$, and the trivial $\Delta$-coloring satisfies
Theorem~\ref{thm:jo-kr}. Again, we assume that $\Delta$ is a suitably
large constant.

\subsection{The Algorithm}
\label{sec:kr-algo}

The algorithm in this case is very similar in structure to that in
\S\ref{sec:randproc}. The only difference is in the modifiers: we
replace Step~(5) of that algorithm by the following:
\begin{enumerate}
\item[5'.] For each pair $\wg$, generate modifiers $M^w_\vg$ for all $v
  \sim w$ as follows:
  \begin{itemize}
  \item If $\wg \not \in \T$ (i.e., $A_\wg = 0$) then $M^w_\vg = 1$
    for all $v \in N(w)$.
  \item Else, if $\wg \in \T$, then let $H = G[N(w)]$ be the graph
    induced on the neighbors of $w$. Use Theorem~\ref{thm:kr-mod} on this
    graph $H$, with $c = \frac14$, and with values $\{ p(v, \g) \}_{v
      \in V(H)}$. This generates modifiers $\Mtilde^w_\vg$ for
    each $v \in V(H) = N(w)$, and then define
    \begin{gather} 
     M^w_\vg = \underbrace{2(1-\eta_\wg)}_{m.o. r.v.}
        \cdot \underbrace{\Mtilde^w_\vg}_{m.o. r.v.} \notag
    \end{gather}
  \end{itemize}
\end{enumerate}

We emphasize that we invoke the procedure in Theorem~\ref{thm:kr-mod}
once for each $\wg \in \T$. A few comments on the new modifiers:
\begin{itemize}
\item By Theorem~\ref{thm:kr-mod}(P3) and the fact that $c = \frac14$,
  each modifier has
  \[ \ts \kappa(M^w_\vg) \leq O\big(r^2 + r \ln \big(\sum_{v \sim w}
  p\vg \big) \,\big)\] as long as $\wg$ is a tentatively chosen pair;
  $\kappa(M^w_\vg) = 0$ if $\wg$ is not tentatively chosen. But the pair
  $\wg$ is tentatively chosen with probability $\theta p_c\wg$, and
  $p_c\wg \neq 0$ implies that the probability is at most $\crowd$. This
  implies that $\kappa(M^w_\vg) \leq O(r^2 + r \ln \ln \Delta)$.
\item The maximum value of the modifer is given by
  Theorem~\ref{thm:kr-mod}(P2), and again using the bound on total
  probability of any neighborhood, we get that
  \begin{gather}
    \khat(M^w_\vg) \leq O(\ln \Delta)^r \cdot O(1)^{r^2}. \label{eq:jo-15}
  \end{gather}
  If we define $\khat := \max_{w,v,g} \khat(M^w_\vg)$, we can use this
  to again define $\pstar := 2\khat \phat$, just the value of $\khat$
  has changed.  In the analysis for the $r$-local-colorability case, we
  used that $\pstar$ was only greater than $\phat$ by a $\khat = 2r \leq
  \Delta^\e$ factor. This time $\khat$ is given by~(\ref{eq:jo-15}), but
  for values $r \leq c'' \sqrt{\ln \Delta}$ for some suitably small
  constant $c''$, $\khat$ is again $\Delta^\e$ and we can use $\pstar
  \leq \phat \Delta^\e$ as before.
\end{itemize}

The runtime: Constructing the modifier requires $\poly(\Delta) \cdot r$ time and
we construct a modifier for each vertex at most $T$ times. Hence, this algorithm runs in 
$O(n\cdot r \cdot \poly(\Delta))$

\subsection{The Altered Parameters}
\label{sec:kr-param}

The analysis remains very similar, we just indicate the changes (in blue). The
parameters change slightly from \S\ref{sec:param}: we now set $c :=
\frac14$ (as mentioned in the algorithm description), and set $b := a -
c - \e$. This affects the value of $K$. 
Moreover, since the modifers $M^w_\vg$ change, the value of 
$C = \max_{w,v,g} \kappa(M^w\vg)$ becomes $O(r^2 +
r \ln \ln \Delta)$ by the discussion above. To summarize, here are the
new parameters; those in blue differ from their counterparts in \S\ref{sec:param}.
\begin{alignat*}{3}
  \e &:= 1/100 & \theta &:= \Delta^{-1/4 + 2\e} & s &:= |L| := \frac{\Delta}{K} \\
  \color{blue} \khat &:= \color{blue} 2^{O(r^2)} \cdot \ln(\Delta)^{O(r)} & \qquad \qquad \phat &:= \Delta^{-3/4 - 5\e} 
        & \color{blue}\qquad\qquad \pstar &:= \color{blue} \khat \phat  \\
  a &:= 1/2 - 3\e &  \color{blue}\cee &:= \color{blue}1/4 & \color{blue}\bee &:= \color{blue}a-\cee -\e \\
  \color{blue}C &:=  \color{blue}O(r^2 + r \ln \ln \Delta) & T &:= \frac{2\e}{a \cdot \theta} \ln \Delta 
    &\qquad\qquad \color{blue}K &:= \color{blue}\frac{(b/4) \e \ln \Delta}{C}
\end{alignat*}

\subsection{The Analysis}
\label{sec:kr-analysis}

Upon closer inspection, the proofs of invariants~\ref{inv:p},
\ref{inv:d} and~\ref{inv:H} go through verbatim, since they do not use
any properties of the modifiers being used.  Only the analysis of
the invariant~\ref{inv:e} bounding the energy needs to be changed. In
particular, Claim~\ref{clm:energy1} no longer holds, and we must use the 
slightly weaker claim below (we will defer the proof to the end of the section).
\begin{claim}
  \label{clm:energy3}
  For any $v$, $\Exp[\nrg_G(v, \pa')] = \sum_{u \in G} \Exp[\nrgy{uv}{\pa'}] 
      \leq \left(1+(c+\e) \theta \right) \nrg_G(v,\pa)$.
\end{claim}

Combining Claims \ref{clm:energy3} and \ref{clm:energy2},
\[ \sum_{u\in G'}\Exp[\nrg(uv,\pa')] \leq \left(1-(a-c-\e)\theta \right)
\sum_{u\in G} \nrg(uv,\pa) = \left(1-\bee\theta \right) \sum_{u\in G}
\nrg(uv,\pa),
\]
since we redefined $b$ to be $a-c-\e$. This is the analog of
(\ref{eq:jo-14}); we now need to show the concentration. And indeed, the
argument in \S\ref{sec:energy-conc} is almost independent of the
modifiers, except for the use of Claim~\ref{clm:energy1} in the proof of
Claim~\ref{clm:energy1whp}; we can now use Claim~\ref{clm:energy3} instead
and get identical results up to changes in the constants (which are absorbed
in the \whpD claims). This proves that invariant \ref{inv:e} also
holds. The rest of the analysis is unchanged for the new modifiers.
Finally, the number of colors used is $O(\frac{\Delta}{K}) =
O(\frac{\Delta \cdot C}{\ln \Delta})$ again; plugging in the value of
$C$ calculated above gives Theorem~\ref{thm:jo-kr}.

It only remains to prove Claim~\ref{clm:energy3}, which we do next, and
to give the construction of the new modifers, which appears in the next
section \S\ref{sec:kr-modifier}.

\begin{proofof}{Claim~\ref{clm:energy3}}
  We begin as in the proof of Claim~\ref{clm:energy1}. Using Fact~\ref{fct:stopped}, 
  we can replace the stopped product with unstopped products, i.e.~for 
  any nodes $u,v \in G$, and color $\gamma \in L$
  \begin{align}
     \Exp[\pa'\ug \pa'\vg] 
    &\leq \pa\ug \pa\vg \Exp\left[ \prod_{w \sim u} M^w_\ug \prod_{x \sim v}
      M^x_\vg\right]. \notag
  \end{align}
    Using independence of the random variables, 
  \begin{gather}
      \Exp \left[ \bigg( \prod_{w\sim u} M^w_\ug \bigg) \bigg(
        \prod_{w\sim v} M^w_\vg\bigg) \right] = \prod_{w\sim 
    v:w\not\sim u} \Exp \left[ M^w_\ug \right] \prod_{w\sim v:w\not\sim
    u} \Exp \left[ M^w_\vg \right] \prod_{w:uvw\in \triangle} \Exp
  \left[ M^w_\ug M^w_\vg \right].  \label{eqn:energy-expan}
  \end{gather}
  The expectations in the first two
  products are~$1$, so focus on the expectations in the last product:
  \begin{align*}
    &\Exp\left[ M_{\ug}^w M_{\vg}^w\right] \\
    &= \Exp\left[ M_{\ug}^w M_{\vg}^w \mid w \notin
      \mathcal{T}_v \right] \cdot \Pr[ w\notin \mathcal{T}_v] +
    \Exp\left[ M_{\ug}^w M_{\vg}^w
    \mid w \in \mathcal{T}_v \right] \cdot
    \Pr[ w\in \mathcal{T}_v]
  \end{align*}
  When $w \notin \mathcal{T}_v$, then all modifiers $\{M^w_\vg\}_v$ have value
  $1$. On the other hand, if $w\in \mathcal{T}_v$, then either $\eta_\wg
  = 1$ (w.p.\ half) and $M^w_\vg M^w_\ug = 0$, or else $\eta_\wg = 0$
  (also w.p.\ half) and $M^w_\vg M^w_\ug = 4\Mtilde^w_\vg
  \Mtilde^w_\ug$. Moreover, $\Pr[w\in \mathcal{T}_v] =\theta
  p_c\wg \leq \theta \pa\wg$. Plugging this into the expression above,
  \begin{align*}
    \Exp\left[ M^w_\ug M^w_\vg \right] &= 
    1\cdot (1 - \theta p_c\wg) + 
    \ts\frac12 \cdot \Exp\left[ 4\Mtilde^w_\vg
      \Mtilde^w_\ug\right] \cdot \theta p_c\wg \\
    &\leq 1 + 
    2 \, \Exp\left[ \Mtilde^w_\vg
      \Mtilde^w_\ug\right] \cdot \theta \pa\wg 
  \end{align*}
  Taking the product over all $w$ that are common neighbors of $u,v$,
  \begin{align*}
   \prod_{w: uvw \in \triangle} & \left(1 + 
      2 \, \Exp\left[ \Mtilde^w_\vg
        \Mtilde^w_\ug\right] \cdot \theta \pa\wg \right) \\
    &= 1 + \sum_{w: uvw \in \triangle} 2 \theta \pa\wg \cdot
      \Exp[\Mtilde^w_\ug \Mtilde^w_\vg] 
      +o(\Delta^{-2-\e}) \\
    &\leq (1 +  \e\theta) + \sum_{w: uvw \in \triangle} 2 \theta \pa\wg \cdot
      \Exp[\Mtilde^w_\ug \Mtilde^w_\vg] 
  \end{align*}
  In the second expression, the second-order terms in the product
  expansion are bounded by $(\theta \cdot \phat \cdot \Delta^{\e})^2 \in
  o(\Delta^{-2-\e})$, since $\pa\vg \leq \phat$ and $\Mtilde^w_\ug \leq
  \khat \leq \Delta^\e$ (by the comments at the end of
  \S\ref{sec:kr-algo}). The subsequent inequality uses that
  $o(\Delta^{-2-\e})\leq \e \theta $ for large enough $\Delta$.

  Now summing over all colors $\g$ and over all $u \sim v$, we get that
  for vertex $v$,
  \begin{align*}
    \Exp[\sum_{u \sim v} \nrg(uv, \pa') ] &\leq \sum_{u \sim v} \sum_\g
    \pa\ug \pa\vg \left[ (1 + \e\theta) + \sum_{w: uvw \in \triangle} 2\theta\pa\wg
      \, \Exp\left[ \Mtilde^w_\vg \Mtilde^w_\ug \right]  \right] \\
    & \leq (1+\e\theta) \nrg_G(v,\pa) + 
    \sum_{\gamma} \sum_{u \sim v} \sum_{w: uvw \in \triangle} 2\theta
    \pa\ug \pa\vg \, \pa\wg\,
    \Exp[\Mtilde^w_\ug \Mtilde^w_\vg] \\
    \intertext{And now summing over all $w \sim v$ instead of only $w:
      uvw \in \triangle$, and interchanging the summations}
    & \leq (1+\e\theta)\nrg_G(v,\pa) +  \sum_{\gamma} \sum_{w \sim v} \theta   
    \pa\wg \left( \sum_{u \sim v} \pa\vg \pa\ug 
      \cdot 2 \Exp[\Mtilde^w_\ug \Mtilde^w_\vg] \right) \\
    \intertext{Applying Theorem~\ref{thm:kr-mod}(P1) on the inner
      summands,}
    & \leq (1+\e\theta)\nrg_G(v,\pa) + \sum_{\gamma} \sum_{w \sim v} \theta \pa\wg 
    \cdot c\, \pa\vg \\
    &\leq (1+(c+\e)\theta) \nrg_G(v,\pa)
    \end{align*}    
    This completes the proof of Claim~\ref{clm:energy3}. 
\end{proofof}

\subsection{Constructing a \texorpdfstring{$K_r$}{K-r}-free Graph Modifier}
\label{sec:kr-modifier}

\begin{theorem}
  \label{thm:kr-mod}
  Given an integer $t \leq r$, a $K_t$-free graph $H$, a constant $c < 1$
  called the ``contraction'' parameter, values $p: V(H) \to [0,1]$, 
  we can construct modifiers $\{M(v)\}_{v \in V}$ which are \morvs
  that satisfy the following properties:
  \begin{OneLiners}
  \item[(P1)]For every vertex $v \in V(H)$, $\sum_{u \sim v} 2
    p(u)p(v)\, \Exp[M(u)M(v)] \leq c \cdot p(v)$.

  \item[(P2)] The maximum value of any $M(v) \leq \left(\frac{8\,
        p(V)}{c}\right)^{t-2}\cdot 16^{\binom{t-2}{2}}$.

  \item[(P3)] $\Exp[M(v) \log M(v)] = O(t^2 + t \log (p(V)/c))$.
  \end{OneLiners}
  Moreover, this construction works in time  $\poly(|V(H)|)\cdot r$.
\end{theorem}

Remark: The graph $H$ should be viewed as the neighborhood of some vertex $w$.

\begin{proof}
  The proof is via induction. The base case is when $t = 2$; the graph
  $H$ is $K_2$-free (i.e., it has no edges), then we return $M_2(v) = 1$
  for all $v$. This ``trivial modifier'' satisfies the properties above.
  Else $t \geq 3$; in this case we recursively build the modifer. Let $V
  = V(H)$ be the vertex set of the graph.
  \begin{enumerate}
  \item Define $q(v) = p(v)/c$. Sample a set $X \sse V$ of vertices using a
    dependent sampling technique of Gandhi et al.~\cite{GKPS} satisfying
    the following properties:
    \begin{enumerate}[(i)]
    \item $|X| \in \{\floor{q(V)}, \ceil{q(V)}\}$,
    \item $ \Pr[v \in X] = q(v)$, and 
    \item $\Pr[ N(v) \notin X] \leq \prod_{u \sim v} (1-q(u))$.
    \end{enumerate}
    Assumption~\ref{ass:well-def} implies that $q(v) \leq 1$ for all $v$, so
    the step is well-defined.
    
  \item Let $m:=|X|$ and let $X= {x_1,x_2,\dots x_m}$. Define $V_0 := V
    - N(X)$ and let $V_i := N(x_i) - (\cup_{j=0}^{i-1} V_j)$. This
    partitions the vertex set $V$ into $m+1$ sets. For $i \in \{1,
    \ldots, m\}$ the induced graphs $H[V_i]$ are $K_{t-1}$-free;
    however, $H[V_0]$ might still contain a $K_{t-1}$.

  \item Let $A:= \left\{i\mid p(V_i) > c\right\}$ and $B:=[1 \ldots m]
    \setminus A$.  Let $\alpha$ be an r.v.\ with $\Pr[\alpha = i] =
    w_i$, where
    \begin{equation*}
      w_i = \begin{cases}
        \frac{3}{4} & \text{ if $i=0$}  \\
        \frac{p(V_i)}{8\,\sum_{i \in A} p(V_i)} & \text{ if $i \in A$}\\
        \frac{1}{8\,|B|} & \text{ if $i \in B$} 
      \end{cases}
    \end{equation*}

  \item For each $i\in \{1,\ldots,m\}$, recursively construct a modifier
    $M'_i$ on the induced $K_{t-1}$-free graph $H_i := H[V_i]$ using
    contraction parameter $c' = \frac12 c$ and values $p'(v) = p(v)/w_i$
    for all $v \in V_i$. (Assumption~\ref{ass:well-def}(ii) implies that the
    new values $p'$ satisfy the requirements of the Theorem.)  Define
    $M'_0$ to be the trivial modifier assigning value $1$ to all $v \in
    V_0$.  Define
    \[ M(v) = \sum_{i=0}^m \one{\alpha = i} \cdot \one{v \in V_i}
    \cdot \frac{1}{w_i} M'_i(v). \] In other words, the modifier
    picks $\alpha$, defines $M(v) = 0$ for all $v \in V
    \setminus V_\alpha$, and returns a scaled-up version of the
    recursively constructed modifier for $V_\alpha$.
  \end{enumerate}

  In order to ensure the algorithm is well-defined, we need some
  assumptions
  \begin{assumption}
    \label{ass:well-def}
    The construction above satisfies:
    \begin{OneLiners}
    \item[(i)] $q(v) = p(v)/c \leq 1$ for all $v \in V$.
    \item[(ii)] $p'(v) = p(v)/w_i \leq 1$ for all $v \in V$.
    \end{OneLiners}
  \end{assumption}
  We now prove that $M$ satisfies properties (P1)--(P3).

\subsubsection{Satisfying property~(P1)}

We need to show that for $v \in V$, we have
\[ \sum_{u \sim v} 2 p(u) p(v) \Exp[ M(u) M(v)] \leq c\cdot p(v). \]
Note that the expectation is over the random choice of $X$, the choice of
$\alpha$, and the internal randomness of the modifiers (denoted as IR).
\begin{align} 
&  \sum_{u \sim v} 2 p(u) p(v)\, \Exp_{X,\alpha,IR}[ M(u) M(v)]  \\
&= \sum_{u \sim v} 2 p(u) p(v)\, \Exp_{X,\alpha,IR} \left[ M(u) M(v) \mid 
            v \in N(X) \right] \, \Pr_X[v \in N(X)] \notag \\
&\qquad  + \sum_{u \sim v} 2 p(u) p(v)\, \Exp_{X,\alpha,IR} \left[ M(u) M(v) 
    \mid v \notin N(X)\right] \, \Pr_X[v \notin N(X)]  \notag\\
&\leq \sum_{u \sim v} 2 p(u) p(v)\, \Exp_{X,\alpha,IR} \left[ M(u) M(v) \mid 
            v \in N(X) \right] \notag \\
&\qquad  + \sum_{u \sim v} 2 p(u) p(v)\, \Exp_{X,\alpha,IR} \left[ M(u) M(v) 
    \mid v \notin N(X)\right] \, \Pr_X[v \notin N(X)] \label{eq:11}
\end{align}

Let us concentrate on the first summand in~(\ref{eq:11}), and condition
on some $X$ such that $v \in N(X)$; 
\begin{align*}
  &\sum_{u \sim v} 2 p(u) p(v) \Exp_{\alpha,IR} \left[ M(u) M(v)
    \mid v \in N(X), X \right]  \\
  &= \sum_{u \sim v} 2 p(u) p(v) \Exp_{\alpha,IR}
     \left[ \sum_{i=1}^m \one{\alpha=i} \cdot  \one{u\in V_i} \frac{1}{w_i} 
         \cdot M'_{i}(u) \cdot \one{\alpha = i} \cdot \one{v \in V_i} \frac{1}{w_i} 
         \cdot M'_{i}(v) \mid v \in N(X), X \right] \\
    \intertext{Since the internal randomness for the modifiers at the next level
        and $\alpha$ are independent, we get}
  &= \sum_{u \sim v} 2 p(u) p(v) \sum_{i=1}^m 
            \Exp_{\alpha}\left[\one{\alpha = i} \mid v \in N(X), X \right]
      \Exp_{IR} \left[  \frac{1}{w_i} \cdot M'_{i}(u) \cdot
          \frac{1}{w_i} \cdot M'_{i}(v) \mid v \in N(X), X \right] \\
      \intertext{Rearranging the sum, and noting that 
          $\Exp_{\alpha}[\one{\alpha = i}\mid X]= \Pr[\alpha=i\mid X] =
          w_i$, we get }
  &= \sum_{i=1}^m w_i \sum_{u\sim v, u \in V_i} 2 \frac{p(u)}{w_i} \,
      \frac{p(v)}{w_i} \, \Exp_{IR} \left[M'_i(u) M'_i(v) \right] \\
  &= \sum_{i=1}^m w_i \sum_{u\sim v, u \in V_i} 2\, p'(u)\,p'(v)\,
       \Exp_{IR} \left[M'_i(u) M'_i(v) \right] \\
      \intertext{Applying the induction hypothesis on $H[V_i]$ with values 
          $p'$,}
  &\leq \sum_{i=1}^m w_i\, c'\, p'(v) \leq \frac{c}{2}\,  p(v)
\end{align*}

We now turn to the second summand in~(\ref{eq:11}). In particular,
consider the expectation 
\begin{align*}
  &   \Exp_{X,\alpha,IR} \left[ M(u) M(v) \mid v \notin N(X) \right] \\
  &= \Pr[\alpha =0] \cdot
  \Exp_{X,IR} \left[ M(u) M(v) \mid v \notin N(X),\, \alpha=0 \right] \\
  &\qquad + \Pr[\alpha \neq 0] \cdot
  \Exp_{X,\alpha,IR} \left[ M(u) M(v) \mid v \notin N(X),\,
    \alpha\neq0 \right] 
\intertext{If we ensure that the value $w_0$ is chosen independently of
  $X$, we get that $\{\alpha = 0\}$ is independent of $X, IR$. Moreover,
since $v \notin N(X)$, it lies in $V_0$. By construction, $M(u)M(v)$
will be non-zero only if $u$ also lies in $V_0$, which causes the second
summand above to disappear, and give
}
  &=\Pr[\alpha =0] \cdot
  \Exp_{X,IR} \left[ M(u) M(v) \mid v \notin N(X) \right] \\
  &= \Pr[\alpha = 0] \,\Exp_{X}[\frac{1}{w_0} \bs{1}(u\notin N(X))
  \frac{1}{w_0}
  \bs{1}(v\notin N(X)) \mid v\notin N(X)] \\
  &= \Pr[\alpha = 0] \, \frac{1}{w_0^2}\, \Pr_X[ u\notin N(X) \mid
  v\notin N(X)] \\
  & \leq \Pr[\alpha = 0] \, \frac{1}{w_0^2} = \frac{1}{w_0}
\end{align*}
This shows that the second summand of~(\ref{eq:11}) is upper bounded by
\begin{align*}
  \sum_{u \sim v} p(u) p(v) \frac{1}{w_0} \,\Pr[v \notin N(X)] &\leq
  \frac{ p(v)}{w_0} \bigg(\sum_{u \sim v} p(u)\bigg) \bigg(
  \prod_{u \sim v} (1- q(u)) \bigg) \\
  &\leq \frac{p(v)}{w_0 }\cdot c \cdot \bigg(\sum_{u \sim v} q(u)\bigg)
  \exp\bigg( -\sum_{u \sim v} q(u)\bigg) \\ 
  &\leq \frac{p(v)}{w_0 } \cdot c \cdot \frac{1}{\mathrm{e}} \leq
  \frac{c}2\, p(v).
\end{align*}
In the first inequality we used the negative correlation property of the
dependent sampling scheme of Gandhi et al., in the second we used the
definition of $q(u) = p(u)/c$, in the third inequality we used that
$x\exp(-x) \leq \frac{1}{\mathrm{e}}$ for all $x$. The final inequality
uses $w_0 = 3/4 \geq 2/\mathrm{e}$. Hence the two summands sum up to at
most $c p(v)$, proving Property~(P1).

\subsubsection{Satisfying Properties (P2)}

\begin{claim}
  \label{clm:wi-pprime}
  The new values satisfy the following:
  \begin{OneLiners}
  \item[(i)] $w_i \geq \frac{c}{8\, p(V)}$ for all $i \in \{1,\ldots, m\}$.
  \item[(ii)] $p'(V_i) \leq 8 p(V)$.
  \end{OneLiners}
\end{claim}

\begin{proof}
  For (i), for $i \in A$, $w_i = \frac{p(V_i)}{8\,\sum_{i \in A} p(V_i)}
  \ge \frac{c}{8\,p(V)}$ by the definition of $A$.  Moreover, for $i \in
  B$, $w_i \geq \frac{1}{8|X|} \geq \frac{c}{8\, p(V)}$. (Here we ignore
  the issues caused by $|X|$ being an integer adjacent to $p(V)$ rather
  than being equal to it.) Note that (i) does not make any claims about $w_0$.

  For (ii), $p'(V_0) = \frac43 p(V_0) \leq 8 p(V)$. For $i \in A$, 
  \[ p'(V_i) = \frac{p(V_i)}{w_i} = 8 \sum_{i \in A} p(V_i) \leq 8\,
  p(V). \] For $i \in B$, $p'(V_i) = \frac{p(V_i)}{w_i} \leq
  \frac{c}{w_i} \leq 8\, p(V)$ by part~(i).
\end{proof}

To prove property~(P2), by the IH the maximum value of the recursively
constructed modifer $M_i$ is 
\[ \left(\frac{8\cdot p'(V_i)}{c'}\right)^{t-3}\cdot 16^{\binom{t-3}{2}}
\leq \left(\frac{8\cdot p(V)}{c}\right)^{t-3} \cdot 16^{t-3} \cdot
16^{\binom{t-3}{2}} = \left(\frac{8\cdot p(V)}{c}\right)^{t-3} \cdot
16^{\binom{t-2}{2}}, \] using the definition of $c'$ and
Claim~\ref{clm:wi-pprime}(ii). If we consider $i \in \{1, \ldots, m\}$,
then scaling up by $1/w_i$ causes the maximum value to be at most
\[ \frac{1}{w_i} \times \left(\frac{8\cdot p(V)}{c}\right)^{t-3}
\cdot 16^{\binom{t-2}{2}} \leq \left(\frac{8\cdot p(V)}{c}\right)^{t-2}
\cdot 16^{\binom{t-2}{2}} \] by Claim~\ref{clm:wi-pprime}(i). If we
consider $i = 0$, then using that $M_0 \equiv 1$, the maximum value is
$4/3$, which is only smaller (since $t \geq 3$).

For property~(P3), observe that $\Exp[ M \log M ] \leq \Exp[ M ] \cdot
\log M_{\max}$, where $M_{\max}$ is the maximum value $M$ takes. But if
$M(v)$ is a \morv then $\Exp[M(v)] = 1$, so $\Exp[ M(v) \log M(v) ]$ is
bounded by the logarithm of the expression in property~(P2).
\end{proof}

\subsubsection{Satisfying the Assumptions}

We still have to address the issue of the validity of the assumptions in
Assumption~\ref{ass:well-def}. We assume we start off with a $K_r$-free
graph with $p(V) \leq O(\log \Delta)$, and
a contraction parameter $c = \frac14$ (say). Let $p^t(v)$ be the
probability values at some stage where the current vertex set is $V^t$
(which is $K_t$-free), then by Claim~\ref{clm:wi-pprime} and algebra.
\begin{OneLiners}
  \item[(a)] $p^t(V^t) \leq 8^{r-t} p(V)$, and
  \item[(b)] for all $v \in V^t$, $p^t(v) \leq p(v) \cdot
    \frac{p(V)^{r-t} 16^{\binom{r-t}{2}}}{c^{r-t}}$.
\end{OneLiners}
Consequently, $p^t(v) \leq p(v) \cdot (p(V)^r \cdot 16^{r^2})$, and if
we start off with $p(v) \leq \pstar$ and $r \ll \sqrt{\log \Delta}$, we
ensure assumption~(ii) that $p^t(v) \leq 1$ for all stages
$t$. Assumption~(i) demands $q^t(v) = \frac{p^t(v)}{c_t} = O(p^t(v) \cdot
2^t) \leq 1$ which is satisfied again by the same conditions.

\appendix

\section{Probabilistic Tools and Useful Lemmas}
\label{sec:tools}

\subsection{Concentration Bounds}

The following large-deviation bound is standard, see, e.g.,~\cite{AS}.
\begin{theorem}[A Large Deviation Bound]
  \label{thm:chernoff}
  For independent $[0,m]$-bounded random variables $X_1, X_2, \ldots,$ with
  $X := \sum_i X_i$ having mean $\Exp[X] \leq \mu$, given any $\lambda
  \geq 0$,
  \[ \Pr\left[ \abs{X  - \Exp[X]} \geq \lambda \right] \leq 2\exp\left\{ -
    \frac{\lambda^2}{m(2\mu + \lambda)} \right\} . \]
  In particular, this probability is at most $1/\poly(\Delta)$ when
  $\lambda = O(\sqrt{\mu m \ln \Delta} + m \ln \Delta)$.
\end{theorem}

For a multilinear polynomial $f(x) = f(x_1, x_2, \ldots, x_n)$ 
with nonnegative coeffecients and degree $\leq q$, and $n$ independent
random variables $Y=(Y_1,Y_2,\dots ,Y_n),$ we define $\mu_r$ 
(for every $r \leq q$) as follows 
    \[ \mu_r = \max_{\substack{S \subseteq [n], |S|=r\\ S=\{s_1,\dots,s_r\} }}
                        \left( \frac{\partial^r f}{ \partial x_{s_1} 
                                \partial x_{s_2} \dots \partial x_{s_r}}
                        \middle|_{\substack{\phantom{i} \\ \Exp[|Y_1|],\dots, 
                                    \Exp[|Y_n|]}} \right) \]
Building on a long line of work, the following bound is presented by Schudy and
Sviridenko~\cite{SS12}.

\begin{theorem}[Large Deviation for Polynomials]
  \label{thm:ss-poly}
  Consider independent $[0,m]$-bounded r.v.s $X_1, X_2,
  \ldots$ and let $X := (X_1, X_2, \ldots, X_n)$.  Let $f(x) = f(x_1,
  x_2, \ldots, x_n)$ be a multilinear polynomial of degree $q$ with
  non-negative coefficients, and let $f(X)$ have moment parameters
  $\mu_0, \mu_1, \ldots, \mu_q$. There exists a universal constant $C =
  C(q)$ such that
  \[ \Pr\left[ \abs{f(X)  - \Exp[f(X)]} \geq \lambda \right] \leq e^2
  \max\left\{ \max_{r = 1,\ldots, q} \exp\left\{ -
    \frac{\lambda^2}{C\cdot m^r \cdot \mu_0\mu_r} \right\} ~,~
  \max_{r = 1,\ldots, q} \exp\left\{ -
    \left( \frac{\lambda}{C\cdot m^r \cdot \mu_r} \right)^{1/r} \right\} \right\}. \]
\end{theorem}

\begin{proof}
  Use Theorem~1.2 of the Schudy-Sviridenko paper~\cite{SS12} and the
  observation that any $[0,m]$-bounded r.v.\ is moment bounded by
  parameter $L = m$.
\end{proof}

\begin{corollary}
  \label{cor:ss-poly-2}
  Consider independent $[0,1]$-bounded r.v.s $X_1, X_2,
  \ldots$ and let $X := (X_1, X_2, \ldots, X_n)$.  Let $f(x) = f(x_1,
  x_2, \ldots, x_n)$ be a multilinear polynomial of degree $2$ with
  non-negative coefficients, and let $f(X)$ have mean $\Exp[f(X)] \leq
  \mu$ and moment parameters $\mu_1, \mu_2 \leq O(1)$. Then
  \[ \Pr\left[ \abs{f(X)  - \Exp[f(X)]} \geq \lambda \right] \leq e^2
  \max\left\{ \exp\left\{ -
    \frac{\lambda^2}{O(\mu)} \right\} ~,~
  \exp\left\{ - O(\lambda)^{1/2} \right\} \right\}. \]
  In particular, this probability is at most $1/\poly(\Delta)$ when
  $\lambda = O(\sqrt{\mu \ln \Delta} + \ln^2 \Delta)$.
\end{corollary}

\subsection{The Lov\'asz Local Lemma}

The following theorem essentially follows from Moser and Tardos~\cite{MT10}.

\begin{theorem}
  \label{thm:mt-lll}
  Consider a set of $n$ independent random variables $\F = \{X_i\}_{i =
    1}^n$, and assume that sampling each r.v.\ from the underlying
  distribution can be done in constant time.  Given a collection of $m$
  subsets $\{ S_j \sse \F \}_{j = 1}^m$ such that the ``bad'' event $\B_j$
  is completely determined by the r.v.s in subset $S_j$, define the degree
  $d_j = \abs{\{ j' \in [m] \mid S_j \cap S_{j'} \neq \emptyset
    \}}$. Define $p_j := \Pr[ \B_j]$. Suppose
  \[ (\max_j p_j) \cdot (\max_j d_j) \leq 1/4 \] then there is an
  algorithm running in time $\poly(m,n)$ to find a setting of the random
  variables $X_i$ such that none of the bad events occur.
\end{theorem}

\subsection{Auxiliary Lemmas}

\begin{lemma}
  \label{lem:p-to-pa}
  Suppose $h(v,p) \geq (1 - \delta) \ln \Delta$ and $p(\P_v) \in (1 \pm
  \nu)$, then $\pa(\P_v) \geq 1 - 6(\delta+\nu)$. If $\nrg_{G}(v,p) \leq
  2K$ also holds, then $\pc(\P_v) \geq 1-6(\delta+\nu)- 2 \e$.
\end{lemma}

\begin{proof}
  First, we prove the bound on $\pa(\P_v)$. 
  Recall that $\pa\vg = p\vg
  \cdot \one{p\vg \leq \phat}$. Since any non-zero probability is at
  least $1/s \geq 1/\Delta$, the entropy
  \begin{align}
    h(v,p) &= - \sum_\g p\vg \ln p\vg
    \leq \bigg( \sum_{\g: p\vg \in (1/s, \phat]} p\vg \bigg) \ln \Delta +
    \bigg( \sum_{\g: p\vg > \phat} p\vg \bigg) \ln 1/\phat \label{eq:jo-12}
  \end{align}
  Let $B := \sum_{\g: p\vg > \phat} \, p\vg$; since $p(\P_v) \in (1 \pm
  \nu)$, we have that \[ \pa(\P_v) = \sum_{\g: p\vg \in [0, \phat]} p\vg
  = \sum_{\g: p\vg \in (1/s, \phat]} p\vg \in (1 - B \pm \nu). \]
  Moreover, $\ln 1/\phat = (\frac34 + 5\e)\ln \Delta$. Finally, by
  assumption, $h(v,p) \geq (1 - \delta) \ln \Delta$.  Substituting
  into~(\ref{eq:jo-12}) and dividing throughout by $\ln \Delta$, and using
  $\e = 1/100$, we get
  \begin{align}
    (1 - B + \nu) + \bigg(\frac34 + 5\e\bigg) B \geq 1 - \delta \quad 
    \implies \quad B \leq 5(\delta + \nu).
  \end{align}
  Hence $\pa(\P_v) \in [1 - 6(\delta + \nu), 1 + \nu]$, which proves the
  first part of the claim.

  Next, the bound on $\pc(\P_v)$. Recall that $\pc\vg = \pa\vg \cdot
  \one{\sum_{u \sim v} p\ug \leq \crowd}$. Since $\e = 1/100$, the
  threshold for zeroing out is $\frac{\ln \Delta}{\e} \geq
  \frac{K}{\e}$. Consequently, if $S_v:=\{ \gamma \mid \sum_{u\sim v}
  p\ug \geq \frac{K}{\e} \}$, then $\pc(\P_v) \geq \pa(\P_v) -
  \sum_{\g \in S_v} p\vg$.  To bound the latter sum, observe that 
  \[ \sum_{\g \in S_v} p\vg \leq \frac{\e}{K} \sum_{\g \in S_v} p\vg
  \sum_{u \sim v} p\ug  \leq \frac{\e}{K} \sum_{u \sim v} \nrg_G(uv, p)
  = \frac{\e}{K} \nrg_G(v,p)
  \leq 2\e.\]
  Hence $\pc(\P_v) \geq 1 - 6(\delta+\nu) - 2\e$
\end{proof}
\fi

\end{document}